\documentclass{article}
\usepackage[utf8]{inputenc}
\usepackage{bbm}
\usepackage{color}
\usepackage{xcolor}
\makeatletter
\newcommand*{\rom}[1]{\expandafter\@slowromancap\romannumeral #1@}
\makeatother
\usepackage{hyperref}
\usepackage{array,booktabs}
\usepackage{wrapfig}
\usepackage{multirow}
\usepackage{tabularx}
\usepackage[english]{babel}
\usepackage{amsmath}
\usepackage{adjustbox}
\usepackage{amsthm}
\usepackage{graphicx}
\usepackage{caption}
\usepackage{subcaption}
\usepackage{cleveref}
\usepackage[labelformat= parens, labelfont={bf,color=black}]{subcaption}
\usepackage[format=hang, textfont=scriptsize, labelsep= period,labelfont={bf,up,color=black}]{caption}
\usepackage{float}
\usepackage{setspace}
\usepackage[space]{grffile}
\usepackage{multicol}
\usepackage{gensymb}
\usepackage{authblk}
\usepackage{scalerel,amssymb,graphicx}
\numberwithin{equation}{section}
\usepackage{tikz}
\usepackage[normalem]{ulem}
\tikzset{mynode/.style={inner sep=2pt,fill,outer sep=0,circle}}
\usepackage[left=1.1in, right=1.2in, top=1.1in,bottom=1.1in]{geometry}
\usetikzlibrary{arrows}
\newcolumntype{M}[1]{>{\centering\arraybackslash}m{#1}}
\tikzset{options/.code={\tikzset{#1}}}
\usepackage{xcolor}
\definecolor{Blue}{rgb}{0.88,1,1}
\definecolor{Gray}{gray}{0.9}
\definecolor{Green}{rgb}{0.0, 0.5, 0.0}
\usetikzlibrary{arrows.meta}
\usetikzlibrary {positioning}
\newtheorem{theorem}{Theorem}[section]
\newtheorem{corollary}{Corollary}[theorem]
\newtheorem{lemma}[theorem]{Lemma}

\theoremstyle{remark}
\newtheorem{remark}{Remark}[section]

\usetikzlibrary{arrows.meta}
\usetikzlibrary {positioning}
\title{A New Perspective on Determining Disease Invasion and Population Persistence in Heterogeneous Environments}
\author[1]{Poroshat Yazdanbakhsh}
\author[]{Mark Anderson}
\author[2]{Zhisheng Shuai}
\affil[1]{\small Department of Mathematics and Computer Science, Rollins College, Winter Park, Florida, 32789, USA}
\affil[2]{\small Department of Mathematics, University of Central Florida, Orlando, Florida, 32816, USA}
\date{}

\begin{document}

\maketitle
\begin{abstract}

We introduce a new quantity known as the network heterogeneity index, denoted by $\mathcal{H}$, which facilitates the investigation of disease propagation and population persistence in heterogeneous environments. Our mathematical analysis reveals that this index embodies the structure of such networks, the disease or population dynamics of patches, and the dispersal between patches. We present multiple representations of the network heterogeneity index and demonstrate that $\mathcal{H}\geq 0$. Moreover, we explore the applications of $\mathcal{H}$ in epidemiology and ecology across various heterogeneous environments, highlighting its effectiveness in determining disease invasibility and population persistence.
\end{abstract}
\section{Introduction}\label{sec:introduction}

Heterogeneous networks have gained significant popularity in various fields, including biology, physics, engineering, and public health, in recent decades. The significance of these networks stems from their ability to incorporate diverse characteristics and behaviors, providing invaluable insights into understanding complex systems. Real-world instances of these networks, such as the internet, airline networks, food webs, and social networks, illustrate how entities like individuals, devices, cities, and countries can be interconnected and engage in interactions with one another.

A heterogeneous network can be modeled as a digraph consisting of vertices (nodes) and directed arcs. In the realm of mathematical biology, vertices represent patches (regions), and directed arcs represent movements between patches through various means such as roads, waterways, and air routes. The structure of heterogeneous networks and the dispersal between patches give rise to mathematical complexities involving Laplacian matrices, which serve as algebraic representations of these networks.

This paper aims to explore how the structure of heterogeneous environments, the dynamics of patches, and the dispersal between patches impact the population or disease evolution. We seek to identify the key factors that facilitate or hinder the propagation of individuals or disease in heterogeneous environments. This work holds significant importance in our rapidly globalizing world,
where not only the distribution and population of species are affected but also the emergence of epidemics and pandemics, such as COVID-19, is influenced.

Given the increasingly interconnected nature of our world, the progression of an infectious disease or a population of species heavily depends on the connectivity of patches. This has captured the attention of ecologists and epidemiologists, who are interested in studying multi-patch disease or population models in heterogeneous environments compromised of connected patches. 
 
The impact of the spatial heterogeneity and rates of movements on population survival and extinction have been investigated in a variety of settings, for example the evolution of dispersal in patchy landscapes
 investigated by~\cite{cantrell2020evolution, hastings1982dynamics, holland2008strong, kirkland2006evolution,roth2014persistence}. Numerous work has been conducted on the population dynamics of a system of stochastic differential equations, exploring both discrete-time and discrete-space in studies done by~\cite{benaim2009persistence, schreiber2010interactive, schreiber2011persistence}, as well as continuous-time and discrete-space work done by~\cite{evans2013stochastic,hening2018stochastic}.

There has also been a considerable amount of work that captures the interplay between dispersal and heterogeneous environments and how this can influence disease spread, see~\cite{allen2007asymptotic, arino2003multi, chen2020asymptotic,chen2022two,eisenberg2013, gao2019travel, kirkland2021impact, lu2023relative} and references therein. Tien et al.~\cite{tien2015disease} approximated the network basic reproduction number $\mathcal{R}_0$, a threshold parameter that determines whether a disease dies out or persists,
  up to the second-order perturbation term for a water-borne disease, such as cholera, in a heterogeneous environment. This approximation, inspired by Langenhop's work~\cite{langenhop1971laurent}, employs the Laurent series approach on the Laplacian matrix of a heterogeneous network. This approximation is the key motivation of our work in this paper.

  Previous studies on disease spread in heterogeneous environments primarily centered around theoretical analysis based on the basic reproduction number $\mathcal{R}_0$, which is the largest eigenvalue of a non-negative matrix (the next generation matrix). In this paper, we introduce a novel perspective to examine the impact of dispersal and environmental heterogeneity  
 on disease dynamics.   Our approach is based on investigating the spectral bound of a (not necessarily nonnegative) matrix. Furthermore, our approach can be applicable to understanding population survival and extinction in ecological models in heterogeneous environments. 

An evolution of species or invasion of disease can be regarded as the stability problem of a Jacobian matrix $J$ of a linearized system at an equilibrium. Particularly, the sign of the spectral bound of $J$,  denoted by $s(J)$, determines whether the equilibrium is stable or not. Biologically, a population or an infectious disease eventually dies out over time if $s(J)<0$, and it persists if $s(J)>0$; see for example~\cite{diekmann2010construction,van2002reproduction} for the connection between $s(J)$ and $\mathcal{R}_0$, and~\cite{cosner1996variability, li2010global, lu1993global} for the stability of population dynamics.

The linearization at the equilibrium in spatial population models often leads to the Jacobian matrix $J$ in the following form
\begin{equation}\label{eq:J}
    J= Q - \mu L,
\end{equation}
where $Q=\mathrm{diag} \{q_{i}\}$ is a diagonal matrix with $q_i$, for $1 \leq i \leq n$, describing the population or disease dynamics of patch $i$ in isolation, $\mu>0$ denotes the dispersal rate, and $L$ denotes an $n \times n$ Laplacian matrix corresponding to the dispersal in the network.

As we will explain later in great detail, $J$ in~\eqref{eq:J} can be perceived as a perturbation matrix of $-L$. Furthermore, by the Perron-Frobenius theory~\cite{berman1994nonnegative, horn2013matrix}, $s(J)$ is an eigenvalue of $J$. Thus, theoretical analysis of $s(J)$ can fall within the scope of the analytical perturbation theory~\cite{kato2013perturbation, katoshort}. 
In this paper, we establish an expansion for $s(J)$ as
\begin{equation}\label{eq:1s(j)}
s(J)= \sum_{i=1}^n q_i \theta_i +\frac{1}{\mu} \sum_{i=1}^n \sum_{j=1}^n q_i \ell^\#_{ij} q_j \theta_j+o\left(\frac{1}{\mu}\right),
\end{equation}
where ${\pmb\theta} =(\theta_1, \dots, \theta_n)^\top$ denotes the positive, normalized right null vector of $L$ and $L^\# = (\ell^\#_{ij})$ denotes the group inverse of $L$. 
The first term in expansion~\eqref{eq:1s(j)} is regarded as the {\it network average}, see~\cite{tien2015disease}, and denoted by
\begin{equation}\label{eq:net-A}
  \mathcal{A}= \sum_{i=1}^n q_i \theta_i. 
\end{equation}
In fact, $\mathcal{A}$ is expressed as the weighted average of the patch dynamics $q_i$, with weights according to the the distribution of  $\theta_i$.
We define the {\it network heterogeneity index} as
\begin{equation}\label{eq:net-H}
    \mathcal{H}= \sum_{i=1}^n \sum_{j=1}^n q_i \ell^\#_{ij} q_j \theta_j.
\end{equation}
Additionally, $\mathcal{H}$ can be represented in the following variational formula
$$\mathcal{H} = \frac{1}{2} \sum_{i=1}^n \sum_{j \neq i}\ell^\#_{ij}(q_i - q_j)^2\theta_j\geq 0,$$
with the equality holding if and only if $q_i = q_j$ for $1\leq i,j \leq n$. In analog to the first and second moments,  $\mathcal{H}$ in~\eqref{eq:net-H} can be interpreted as covariance.

The remainder of this paper is organized as follows: Section~\ref{sec:preliminary} gives a brief summary of notation and concepts that are used throughout the paper. Section~\ref{sec:perturbation theories} provides rigorous mathematical analysis to investigate the effect of perturbing a singular and irreducible M-matrix on its principal eigenvalue. Section~\ref{sec:networkheterogenity} introduces the network heterogeneity index and discusses some of its fundamental properties. 
 Sections~\ref{sec:application} and~\ref{sec:application2}
illustrate the significance of the network heterogeneity index in determining disease invasion and population persistence in epidemiology and ecology over various heterogeneous network structures. Section~\ref{sec:conclusion} provides concluding remarks.

\section{Notation and Terminology}\label{sec:preliminary}
In this section, we present a brief overview of notation and the well-established results in the literature that form the basis of our investigation. Throughout this paper, unless explicitly specified, we consider all matrices to be real-valued and $n$ dimensional. 

 Let $A=(a_{ij})$ be a matrix and $\sigma(A)$ be the set of eigenvalues of $A$. Denote $s(A)$ as the \textit{spectral bound} of $A$ given by
$$s(A)= \max \{ \mathrm{Re} (\lambda): \lambda \in \sigma(A)\}.$$
For our purpose, let $M=(m_{ij})$ denote the matrix corresponding to the movements between patches in a heterogeneous network of $n$ patches, where $m_{ij}\geq 0$  represents the movement rate from patch $j$ to $i\not=j$ and $m_{ii}=0$. The Laplacian matrix $L$ associated to the movement matrix $M$ is  defined as
\begin{equation}\label{eq:Lap-mat}
 L = (\ell_{ij})=\begin{pmatrix}
  \sum \limits_{j \neq 1} m_{j1} & -m_{12} & \cdots & -m_{1n}\\
 -m_{21} & \sum \limits_{j \neq 2} m_{j2} & \cdots & -m_{2n}\\
 \vdots & \vdots & \ddots& \vdots\\
 -m_{n1} & -m_{n2} & \cdots & \sum \limits_{j \neq n} m_{jn}
\end{pmatrix}.
\end{equation}
The diagonal entry $(i,i)$ of~\eqref{eq:Lap-mat} tracks all the movements in the network to patch $i$, and 
 off-diagonal entry $(i,j)$, $ i \neq j$, describes the movement from patch $j$ to patch $i$. 
Note that $L$ is irreducible (as $M$ is irreducible) and each of its columns sum equals $0$. Thus, $L$ is singular with the algebraic multiplicity of $0$ being simple, and hence $\mathrm{dim}(\mathrm{ker}(L))=1$, see~\cite{ fiedler1973algebraic,merris1994laplacian}. 
As a result, $\mathbf{1}^\top = (1)_{1 \times n}$, the vector of all ones on its entries, is the left null vector of $L$. Let $\pmb{\theta}=(\theta_1, \dots, \theta_n)^\top$ denote the positive right null vector of $L$, normalized so that $\mathbf{1}^\top \pmb{\theta} =1$. Thus,
\begin{equation*}
 \quad \mathbf{1}^\top L =0 \quad \quad \mathrm{and} \quad \quad  \quad L \pmb{\theta} =0.
\end{equation*}
Notice that since $-L$ is irreducible with $-\ell_{ij}\geq 0$, for $i \neq j$, by the Perron-Frobenius theory, see~\cite{berman1994nonnegative, horn2013matrix,yang2010further},
\begin{equation}\label{eq:per-root-neg-L}
 s(-L)=0,
\end{equation}
is the largest eigenvalue of $-L$, commonly referred to as the {\it Perron root} of $-L$, and thus $\mathbf{1}^\top$ and $\pmb{\theta}$ are the left and normalized right {\it Perron vectors} of $-L$, respectively. 

It follows from~\cite{kirkland2013group} that while $L$ is not invertible there exists a unique matrix $L^\#:=(\ell^\#_{ij})$, called the \textit{group inverse} of $L$, which satisfies the following properties:

\begin{equation}\label{groupinverse-prop}
  (a) \quad LL^\# = L^\# L,\qquad (b) \quad LL^\#L = L,  \qquad \mathrm{and}\qquad (c) \quad L^\# L L^\# = L^\#.
\end{equation}
In addition, as demonstrated in~\cite{brualdi1991combinatorial,meyer1978singular}, the connection between $L$ and $L^\#$ is given by
\begin{equation}\label{LL-sharp}
    LL^\# = I_n - \pmb{\theta}\mathbf{1}^\top,
\end{equation}
where $I_n$ denotes the identity matrix.
For further background material on generalized inverses see~\cite{ ben2003generalized,campbell2009generalized}, and on the use of group inverse in particular, see~\cite{kirkland2013group}. Lemma~\ref{lem1} will be used repeatedly in later sections of this paper.
\begin{lemma}\label{lem1}
Let $L$ be an irreducible Laplacian matrix and $L^\#$ the corresponding group inverse. Then,
\begin{itemize}
    \item[(i)] $\mathbf{1}^\top L^\# =\mathbf{0}$.
    \item[(ii)] For every vector $\mathbf{x} \in \ker(L)$, $L^\# \mathbf{x} = L\mathbf{x}=\mathbf{0}$.
    \item[(iii)] For every vector $\mathbf{y} \in \mathrm{range}(L)$,
    $LL^\# \mathbf{y} =L^\#L \mathbf{y} =\mathbf{y}$.
\end{itemize}
\end{lemma}
\begin{proof} Part (i) follows from $\mathbf{1}^\top L =\mathbf{0}$ and~\eqref{groupinverse-prop}, since $\mathbf{1}^\top L^\#= \mathbf{1}^\top L^\# L L^\# = \mathbf{1}^\top L L^\# L^\# = \mathbf{0}$. To prove part (ii), assume $\mathbf{x} \in \mathrm{ker}(L)$. Since $L \mathbf{x} =\mathbf{0}$, it now follows from~\eqref{groupinverse-prop} that $$L^\#\mathbf{x}= L^\#LL^\# \mathbf{x}= L^\#L^\#L\mathbf{x}=\mathbf{0}.$$ 
Finally, suppose $\mathbf{y} \in \mathrm{range}(L)$. Then, there exists a vector $\mathbf{z} \in \mathbb{R}^n$ such that $\mathbf{y}=L\mathbf{z}$. Applying~\eqref{groupinverse-prop} gives
$$\mathbf{y} =L \mathbf{z} = LL^\#L \mathbf{z}= LL^\# \mathbf{y}= L^\# L \mathbf{y},$$ 
and thus part (iii) holds.
\end{proof}

Next, we derive the group inverse $L^\#$ from the Laplacian matrix $L$. Following the method description in~\cite[Section $4$]{kirkland2021impact}, $L$ can be partitioned as
$$
L=\left(\begin{array}{c|c}
\mathbf{\bar{1}}^\top \mathbf{w}&{- \mathbf{\bar{1}}^\top} B\\[2pt] \hline\\[-12pt]
 -\mathbf{w}  & B
\end{array}\right),
$$
\noindent where $B$ is an $(n-1) \times (n-1)$ invertible matrix, obtained from eliminating the first row and column of $L$, $\mathbf{\bar{1}}$ is an $(n-1) \times 1$ vector of all ones, and $\mathbf{w}=\dfrac{1}{\theta_1}B \pmb{\bar{\theta}}$, where $\theta_1$ is the first entry of $\pmb{\theta}$ and $\pmb{\bar{\theta}}=(\theta_2, \dots, \theta_n)^\top$. It follows from~\cite[Corollary $2.3.4$]{kirkland2013group} that
\begin{equation}\label{eq:Lsharp-partition}
L^\# = (\mathbf{\bar{1}}^\top B^{-1} \mathbf{\bar{\pmb{\theta}}}) \pmb{\theta} \mathbf{1}^\top + \left(\begin{array}{c|c}
0&-\theta_1 \bar{\mathbf{1}}^\top B^{-1}\\[2pt] \hline\\[-10pt]
-B^{-1}\pmb{\bar{\theta}}& B^{-1} - B^{-1}\pmb{\bar{\theta}}\bar{\mathbf{1}}^\top - \pmb{\bar{\theta}}\bar{\mathbf{1}}^\top B^{-1}
\end{array}\right).
\end{equation}

In Section~\ref{sec:perturbation theories}, we employ perturbation analysis theory and the properties of group inverse $L^\#$ to investigate how variations in the entries of $-L$ impact the value of its Perron root $0$.

\section{Perron Root of Perturbed $-L$}\label{sec:perturbation theories}
Let $r$ denote the spectral bound of the Jacobian matrix $J=Q - \mu L$ in~\eqref{eq:J}, given by
\begin{equation}\label{s(j)}
r=s(Q - \mu L),   
\end{equation}
where $Q=\mathrm{diag}\{q_i\}$, $\mu>0$, and $L$ is an irreducible Laplacian matrix. By the Perron-Frobenius theory~\cite{berman1994nonnegative, horn2013matrix}, $r$ is the Perron root of $Q - \mu L$.
Letting $\varepsilon=\dfrac{1}{\mu}$ and $\lambda = r \varepsilon$, we get
\begin{equation}\label{eq:lambda}
  \lambda= \varepsilon s(Q- \mu L) = s(\varepsilon Q - L).
 \end{equation}
Thus, $\lambda$ is the Perron root of $\varepsilon Q - L$. Denote $\pmb{\nu}$ as corresponding the positive normalized right Perron vector of $s(\varepsilon Q -L)$. By analytical perturbation theory~\cite{kato2013perturbation, katoshort}, these can be expanded as
$$\lambda=\lambda_0 + \varepsilon \lambda_1 +\varepsilon^2 \lambda_2 + \cdots \quad \mathrm{and} \quad \pmb{\nu} = \pmb{\nu}_0 +\varepsilon \pmb{\nu}_1 +\varepsilon^2 \pmb{\nu}_2+ \cdots.$$
In this section, we use the expansion for $\lambda$ to derive an expansion for $r$.

\begin{lemma}\label{lem:lam-nu}
  Suppose $Q=\mathrm{diag}\{q_i\}$, $\varepsilon>0$,  $L$ is an irreducible Laplacian matrix, and $\pmb{\theta}$ is the normalized right Perron vector of $-L$. If $\lambda$ is the Perron root of $\varepsilon Q - L$ and $\pmb{\nu}$ the corresponding positive normalized right Perron vector
where
\begin{equation}\label{eq:pert-expn}\lambda=\lambda_0 + \varepsilon \lambda_1 +\varepsilon^2 \lambda_2 + \cdots \quad \quad \mathrm{and} \quad \quad \pmb{\nu} = \pmb{\nu}_0 +\varepsilon \pmb{\nu}_1 +\varepsilon^2 \pmb{\nu}_2+ \cdots .
\end{equation}
Then,
\begin{itemize}
\item[(i)]
\hspace{30pt} $\displaystyle \lambda_0=0, \quad \lambda_k =\mathbf{1}^\top  Q\pmb{\nu}_{k-1}-\sum_{i=1}
^{k-1}\lambda_i \mathbf{1}^\top \pmb{\nu}_{k-i}, \quad \mathrm{for} \quad k \geq 1, $

and
\item[(ii)]
\hspace{30pt} $\displaystyle \pmb{\nu}_0=\pmb{\theta}, \quad L\pmb{\nu}_k  = Q \pmb{\nu}_{k-1}-\lambda_k\pmb{\theta} - \sum_{i=1}
^{k-1}\lambda_i \pmb{\nu}_{k-i}, \quad \mathrm{for} \quad k \geq 1$
\end{itemize}
\end{lemma}
\begin{proof}
Since $(\varepsilon Q-L) \pmb{\nu} = \lambda \pmb{\nu}$,
\begin{equation}\label{eq:pert-eqn}
 (\varepsilon Q -  L)(\pmb{\nu}_0+\varepsilon \pmb{\nu}_1 +\varepsilon^2 \pmb{\nu}_2+ \cdots )= (\lambda_0 + \varepsilon \lambda_1 +\varepsilon^2 \lambda_2 + \cdots)(\pmb{\nu}_0 +\varepsilon \pmb{\nu}_1 +\varepsilon^2 \pmb{\nu}_2+ \cdots).  
\end{equation}
Comparing terms at order zero gives
\begin{equation}\label{eq:e0}
-L \pmb{\nu}_0= \lambda_0 \pmb{\nu}_0.   
\end{equation}
By the assumption, $\lambda$ is the Perron root of $\varepsilon Q- L$, and thus $\lambda_0$ is the Perron root of $-L$. It follows from \eqref{eq:per-root-neg-L} that \begin{equation*}
\lambda_0=0 \quad \mathrm{and} \quad \pmb{\nu}_0 = \pmb{\theta}.
\end{equation*}
\noindent Comparing terms at order one in~\eqref{eq:pert-eqn} and applying $\pmb{\nu}_0=\pmb{\theta}$ yields
\begin{equation}\label{eq:lambda1-nu1}
Q\pmb{\theta}-L\pmb{\nu}_1 =\lambda_1 \pmb{\theta}.
\end{equation}
Since $\mathbf{1}^\top \pmb{\theta}=1$ and $\mathbf{1}^\top L=0$, multiplying both sides of~\eqref{eq:lambda1-nu1} by $\mathbf{1}^\top$ gives 
\begin{equation}\label{eq:lambda1}
    \lambda_1 = \mathbf{1}^\top Q \pmb{\theta}.
\end{equation}
Similarly, comparing terms in~\eqref{eq:pert-eqn} at order $k$, for $k \geq 2$, in~\eqref{eq:pert-eqn} yields
\begin{equation}\label{eq:lambdak-nuk}
    Q\pmb{\nu}_{k-1}-L\pmb{\nu}_k = \lambda_k\pmb{\theta}+\sum_{i=1}
^{k-1}\lambda_i \pmb{\nu}_{k-i},\end{equation}
and multiplying both sides of~\eqref{eq:lambdak-nuk} by $\mathbf{1}^\top$ from the left side results in
\begin{equation*}
 \lambda_k =\mathbf{1}^\top  Q\pmb{\nu}_{k-1}-\sum_{i=1}
^{k-1}\lambda_i \mathbf{1}^\top \pmb{\nu}_{k-i}, \quad \mathrm{for} \quad k \geq 2.
\end{equation*}
Finally, we derive $\pmb{\nu}_k$ terms in~\eqref{eq:pert-eqn}, for $k\ge 1$, by arranging~\eqref{eq:lambda1-nu1} and~\eqref{eq:lambdak-nuk} as
\begin{align*}
   L\pmb{\nu}_1 &=Q\pmb{\theta}-\lambda_1 \pmb{\theta},\\
    L\pmb{\nu}_k & = Q \pmb{\nu}_{k-1}-\lambda_k\pmb{\theta} - \sum_{i=1}
^{k-1}\lambda_i \pmb{\nu}_{k-i}, \quad \mathrm{for} \quad k \geq 2,\end{align*}
thus concluding the proof.
\end{proof}
It can be seen from Lemma~\ref{lem:lam-nu} that the computation of $\pmb{\nu}_k$ and $\lambda_k$ depend recursively on the values of $\pmb{\nu}_i$ and $\lambda_i$ for $1\leq i \leq k-1$, making the computation of the higher terms complicated.

The following result derives $\lambda$ up to the second-order perturbation term for $\lambda$ and $\pmb{\nu}$ up to the first-order perturbation term when the diagonal matrix $Q$ is not a scalar multiple of the identity matrix $I_n$. Furthermore, it examines the particular scenario when $Q$ is a scalar multiple of $I_n$.

\begin{theorem}\label{thm:main} Suppose $Q=\mathrm{diag}\{q_i\}$, $\varepsilon>0$,  $L$ is an irreducible Laplacian matrix, and $\pmb{\theta}$ is the normalized right Perron vector of $-L$.
If $\lambda$ isthe Perron root of $\varepsilon Q -L$ and $\pmb{\nu}$ the corresponding right Perron vector, then the following statements hold:

\begin{itemize}

\item[(i)] If $Q$ is not a scalar multiple of $I_n$, then, for some $\alpha \in \mathbbm{R}$,
\begin{equation*}
    \lambda= (\mathbf{1}^\top Q \pmb{\theta} )\varepsilon+ (\mathbf{1}^\top Q L^\#Q\pmb{\theta}) \varepsilon^2+o(\varepsilon^2) \quad \mathrm{\it and} \quad \pmb{\nu}=\pmb{\theta}+ ( \alpha I_n + L^\#Q) \pmb{\theta}\varepsilon+o(\varepsilon)
\end{equation*}

\item[(ii)] If $Q$ is a scalar multiple of $I_n$, that is, $Q= qI_n$ for some $q \in \mathbbm{R}$, then
\begin{equation*}\label{eq:lambda1-new}
    \lambda= q\varepsilon \quad \mathrm{\it and} \quad \pmb{ \nu} =\pmb{\theta}.
\end{equation*}

\end{itemize}
\end{theorem}
\begin{proof}

(i)  It follows from Lemma~\ref{lem:lam-nu} that $\pmb{\nu}_0=\pmb{\theta}$. Also, $\pmb{\nu}_1 = \mathbf{u}+\mathbf{w}$, where $\mathbf{u} \in \mathrm{ker}(L)$ and $\mathbf{w} \in \mathrm{range}(L)$, since $\mathrm{ker}(L)\bigoplus \mathrm{range}(L)=\mathbbm{R}^n$, by~\cite{ben2003generalized}. Since $\mathrm{dim}(\mathrm{ker}(L))=1$, see~\cite{ fiedler1973algebraic,merris1994laplacian}, by the Perron Frobenius theorem, $\mathbf{u}= \alpha \pmb{\theta}$ for some $\alpha \in \mathbbm{R}$, and thus
\begin{equation}\label{eq:newnu1}
 \pmb{\nu}_1= \alpha \pmb{\theta} +\mathbf{w}.
\end{equation} 
By Lemma \ref{lem:lam-nu}(ii),
\begin{equation}\label{Lv}
   L \mathbf{w}= L\pmb{\nu}_1= Q \pmb{\theta} - \lambda_1 \pmb{\theta}.
\end{equation}
Multiplying the above equation from the left side by $L^\#$ and applying Lemma~\ref{lem1}(ii)-(iii) give 
    $$\mathbf{w}= L^\# L\mathbf{w} = L^\# Q \pmb{\theta} - \lambda_1 L^\# \pmb{\theta} = L^\#Q\pmb{\theta}.$$
Substituting $\mathbf{w}$ into~\eqref{eq:newnu1} yields
\begin{equation*}
    \pmb{\nu}_1 = \alpha \pmb{\theta}+ L^\#Q\pmb{\theta}.
\end{equation*}
Using Lemma \ref{lem:lam-nu}(i) to derive $\lambda_2$, we get 
\begin{equation}\label{eq:lambda2}
 \lambda_2= \mathbf{1}^\top Q \pmb{\nu}_1 - \lambda_1 \mathbf{1}^\top \pmb{\nu}_1=\alpha \mathbf{1}^\top Q \pmb{\theta} + \mathbf{1}^\top Q L^\# Q \pmb{\theta}  - \alpha \lambda_1 \mathbf{1}^\top \pmb{\theta} - \mathbf{1}^\top L^\# Q \pmb{\theta}.
\end{equation} 
Applying Lemma~\ref{lem1}(i), $\lambda_1 = \mathbf{1}^\top Q \pmb{\theta}$ from Lemma~\ref{lem:lam-nu}, and $\mathbf{1}^\top \pmb{\theta}=1$ gives
 $$\lambda_2 = \mathbf{1}^\top QL^\# Q\pmb{\theta}.$$

(ii) Since $Q=q I_n$ and $\mathbf{1}^\top \pmb{\theta}= 1$, by Lemma \ref{lem:lam-nu}(i),
\begin{equation}\label{eq:lam11}
   \lambda_1 =\mathbf{1}^\top Q \pmb{\theta} =\mathbf{1}^\top (q I_n) \pmb{\theta} = q.
\end{equation}
Additionally, by Lemma \ref{lem:lam-nu}(ii), $L \pmb{\nu}_1 =Q \pmb{\theta} - \lambda_1 \pmb{\theta} =(q I_n)\pmb{\theta} - q \pmb{\theta}=\mathbf{0}$. That is, $ \pmb{\nu}_1 \in \mathrm{ker}(L)$, and thus
\begin{equation}\label{eq:nu11}
\pmb{\nu}_1=\alpha_1 \pmb{\theta} \quad \text{for some $\alpha_1 \in \mathbbm{R}$},
\end{equation}
since $\mathrm{dim}(\mathrm{ker}(L))=1$.
Likewise, by Lemma \ref{lem:lam-nu},
\begin{align*}
    \lambda_2 &= \mathbf{1}^\top Q\pmb{\nu}_1 - \lambda_1 \mathbf{1}^\top \pmb{\nu}_1=\mathbf{1}^\top (qI_n)\alpha_1 \pmb{\theta} - q \mathbf{1}^\top \alpha_1 \pmb{\theta}=0,\text{ and}\\
    L\pmb{\nu}_2 &=Q \pmb{\nu}_1 - \lambda_2 \pmb{\theta} - \lambda_1 \pmb{\nu}_1 =(qI_n )\alpha_1 \pmb{\theta} - 0 \pmb{\theta} - q \alpha_1 \pmb{\theta} =\pmb{0},
\end{align*}
which implies that 
$\pmb{\nu}_2= \alpha_2 \pmb{\theta}$  for some $\alpha_2 \in \mathbbm{R}$.
 Similar to the above arguments, it can be verified that for $k \ge 2$, $\lambda_k=0$ and $\pmb{\nu}_k=\alpha_k \pmb{\theta}$ where $\alpha_k \in \mathbbm{R}$. Applying~\eqref{eq:lam11} and~\eqref{eq:nu11} to~\eqref{eq:pert-expn} yield
 \begin{equation}\label{eq:lamb-nu-ident}
  \lambda= q\varepsilon \quad \mathrm{and} \quad \pmb{\nu}=\pmb{\theta}+\sum_{k=1} ^\infty\alpha_k \pmb{\theta} \varepsilon^k=\pmb{\theta}(1+\sum_{k=1}^\infty\alpha_k \varepsilon^k).   
 \end{equation}
Since $\mathbf{1}^\top \pmb{\nu}=1$ and $\mathbf{1}^\top \pmb{\theta} =1$, multiplying $\pmb{\nu}$ in~\eqref{eq:lamb-nu-ident} by $\mathbf{1}^\top$ gives
$\sum_{k=1}^\infty\alpha_k \varepsilon^k=0$, and thus $\pmb{\nu} = \pmb{\theta}$.
\end{proof}
\begin{corollary}\label{rem:lambda} Suppose $Q=\mathrm{diag}\{q_i\}$, $\varepsilon>0$,  $L$ is an irreducible Laplacian matrix, and  $\pmb{\theta}= (\theta_1, \dots, \theta_n)^\top$ is the normalized right  Perron vector of $-L$.   The Perron root of $\varepsilon Q -L$ can be written as
\begin{equation*}\label{eq:lambda-new2}
    \lambda = (\sum_{i=1}^n q_i \theta_i )\varepsilon +  (\sum_{i=1}^n \sum_{j=1}^n q_i  \ell_{ij}^\#  q_j\theta_j )\varepsilon^2 + o(
    \varepsilon^2),
\end{equation*}
where $L^\#=(\ell^\#_{ij})$.
\end{corollary}
We are now prepared to present our main result, that is, an expansion for $r= s(Q - \mu L)$ as in~\eqref{s(j)}. Furthermore, we utilize the following result to establish the sharp bounds for $r$, see~\cite{altenberg2012resolvent,chen2022two} for the proof.
\begin{lemma}\label{mon-conv}
Let $r=s(Q-\mu L)$ be defined as in~\eqref{s(j)}. Then,
\begin{equation}\label{eq:monotone}
  \dfrac{d r}{d \mu} \leq 0 \quad and \quad  \dfrac{d^2r}{d\mu^2}\ge 0, 
\end{equation}
with equalities holding in~\eqref{eq:monotone} if and only if $Q$ is a scalar multiple of $I_n$.
\end{lemma}

\begin{theorem}\label{them:appr-r}
Let $r = s(Q - \mu L)$ be defined as in~\eqref{s(j)}. Then, the following statements hold:
\begin{itemize}
    \item[(i)] If $Q$ is not a scalar multiple of $I_n$, then

\begin{equation} \label{spect-expand}
r = \sum_{i=1}^n  q_i \theta_i + \frac{1}{\mu}\sum_{i=1}^n \sum_{j=1}^n q_i  \ell_{ij}^\# q_j \theta_j + o\left(\frac{1}{\mu}\right).
\end{equation}
\noindent Additionally, the sharp bounds of $r$ with respect to $\mu$ are given by
\begin{equation}\label{low-up-bound}
\sum_{i=1}^n  q_i \theta_i \le r \le  \max_i \{q_i\}. 
\end{equation} 
    \item[(ii)] If $Q$ is a scalar multiple of $I_n$, that is $q_1=\dots=q_n=q$ for $q \in \mathbbm{R}$, then
\begin{equation*}\label{eq:r1}
    r= q.
\end{equation*}
\end{itemize}
\end{theorem}

\begin{proof}

\noindent (i) Recall that $r=\frac{\lambda}{\epsilon}$ and $\epsilon = \frac{1}{\mu}$. It follows from Corollary \ref{eq:lambda-new2} that 
\begin{equation*}
     r=\sum_{i=1}^n q_i \theta_i + \varepsilon \sum_{i=1}^n \sum_{j=1}^n q_i  \ell_{ij}^\#  q_j\theta_j + o(
    \varepsilon) = \sum_{i=1}^n q_i \theta_i + \dfrac{1}{\mu} \sum_{i=1}^n \sum_{j=1}^n q_i  \ell_{ij}^\#  q_j\theta_j + o\left(\frac{1}{\mu}
   \right).
\end{equation*}
By Lemma~\ref{mon-conv}, $r$ is non-increasing with respect to $\mu$. Thus, the sharp lower and upper bounds of $r$ occur as $\mu$ approaches $\infty$ and $0$, respectively. In particular, by~\eqref{s(j)}, the sharp lower bound is given by
$$\lim_{\mu \rightarrow 0}s(Q-\mu L)=s(Q)=\max_{
i}\{ q_i\},$$
and the sharp upper bound follows from~\eqref{spect-expand} as
$$\lim_{\mu \rightarrow \infty}\left(\sum_{i=1}^n q_i \theta_i + \frac{1}{\mu}\sum_{i=1}^n \sum_{j=1}^n q_i  \ell_{ij}^\#  q_j\theta_j + o\left(
    \frac{1}{\mu}\right)\right)=\sum_{i=1}^n q_i \theta_i.$$

\noindent (ii) Since $Q=q I_n$, it follows from Theorem~\ref{thm:main} that $\lambda= q\varepsilon $, and thus $r= q$.
\end{proof}

 The expansion presented in~\eqref{spect-expand} can be regarded as an aggregation of within-patch characteristics and between-patch connections. Specifically, the first term in the approximation for $r=s(Q - \mu L)$ represents the weighted average of patch dynamics, while the second term describes the deviations away from the average, which will be discussed in great detail in the following section.


\section{Network Heterogeneity Index: Definition and Properties}\label{sec:networkheterogenity}


Consider a heterogeneous network of $n$ patches with the irreducible Laplacian matrix $L$ defined as in~\eqref{eq:Lap-mat}. Let $\mathbf{q} = (q_1, \dots,q_n)^\top$ and $Q = \mathrm{diag}\{q_i\}:=\mathrm{diag}(\mathbf{q}) $ denote the vector and diagonal matrix that describe the dynamics of patch $i$ for $1\leq i \leq n$.
Define the \textit{network average} $\mathcal{A}$ as
   \begin{equation}\label{A}
    \mathcal{A}:= \sum_{i=1}^n q_i \theta_i =\mathbf{q}^\top \pmb{\theta} ,
\end{equation}
and the \textit{network heterogeneity index} $\mathcal{H}$ as
\begin{equation}\label{H}
    \mathcal{H}:= \sum_{i=1}^n \sum_{j=1}^n q_i  \ell_{ij}^\#  q_j\theta_j = \mathbf{q}^\top L^\# \mathrm{diag}(\pmb{\theta}) \mathbf{q} .
\end{equation}
Recall that $\pmb{\theta}=(\theta_1, \dots, \theta_n)^\top$ denotes the normalized right null vector of $L$, $\mathrm{diag}(\pmb{\theta})$ consists of $\theta_i$ on its diagonal, and $L^\#$ is the group inverse of $L$.


\begin{lemma}\label{rem:row-col-zero} Let $L^\#$ be the group inverse of an irreducible Laplacian matrix $L$ and $\pmb{\theta} =(\theta_1, \dots, \theta_n)^\top$ be its associated normalized right null vector. Then $L^\# \mathrm{diag}(\pmb{\theta})$ has all its row and column sum equal to zero, i.e., for  $1\leq i, j \leq n$,
\begin{equation}\label{eqn:row-col}
\sum_{i=1}^{n} \ell_{ij}^\# \theta_j
=\sum_{j=1}^{n} \ell_{ij}^\# \theta_j = 0 .
\end{equation}
\end{lemma}

\begin{proof} It is evident from Lemma~\ref{lem1}(i)-(ii) that  
$$\mathbf{1}^\top (L^\# \mathrm{diag}(\pmb{\theta}))=0 \qquad \mathrm{and} \qquad (L^\# \mathrm{diag}(\pmb{\theta})) \mathbf{1}^\top =L^\# \pmb{\theta}=0,$$
giving the results required.
\end{proof}

The following result presents a variational formula for the network heterogeneity index $\mathcal{H}$, highlighting its biological significance in characterizing differences among various patches.

\begin{theorem}\label{new-H} Let the network heterogeneity index $\mathcal{H}$ be given as in~\eqref{H}. Then,  
\begin{equation}\label{eq-H2}
\mathcal{H} = -\frac{1}{2}\sum_{i=1}^n \sum_{j\not= i} \ell_{ij}^\# \theta_j(q_i-q_j)^2.
\end{equation}
\end{theorem}

\begin{proof} It follows from~\eqref{H} that
\begin{equation*}
\mathcal{H}=\sum_{i=1}^{n}\sum_{j=1}^{n}  q_i \ell_{ij}^\#  \theta_j q_j=\sum_{i=1}^{n}\ell_{ii}^\# \theta_i q_i^2  + \sum_{i=1}^{n}\sum_{j \neq i}q_i \ell_{ij}^\# \theta_j q_j.
\end{equation*}
By Lemma~\ref{lem1}(i), the column sum of $L^\#=(\ell^\#_{ij})$ is zero, that is $\ell^\#_{ii}= -\sum_{i \neq j}\ell^\#_{ji}$. Thus, the above equation becomes
\begin{gather}\label{eq:H-quad2}
\begin{aligned}\mathcal{H}= &- \sum_{i=1}^{n}\sum_{j \neq i} \ell^\#_{ji} \theta_i  q_i^2+ \sum_{i=1}^{n}\sum_{j \neq i}\ell_{ij}^\# \theta_j q_i q_j\\
&= -\dfrac{1}{2} \sum_{i=1}^{n}\sum_{j \neq i}(\ell^\#_{ji} \theta_i q_i^2  - \ell_{ij}^\# \theta_j q_i^2) -\dfrac{1}{2}\sum_{i=1}^{n}\sum_{j \neq i}\ell_{ij}^\# \theta_j(q_i^2 - 2q_i q_j + q_j^2) \\
& = -\dfrac{1}{2}\sum_{i=1}^{n}q_i^2\sum_{j \neq i}(\ell^\#_{ji} \theta_i  - \ell^\#_{ij}\theta_j) -\dfrac{1}{2} \sum_{i=1}^{n}\sum_{j \neq i} \ell_{ij}^\# \theta_j( q_i -q_j)^2.
\end{aligned}
\end{gather}
Notice that $\sum_{j \neq i} ( \ell_{ji}^\# \theta_i - \ell_{ij}^\# \theta_j) =0,$ which follows directly from  \eqref{eqn:row-col} in Lemma~\ref{rem:row-col-zero}. 
Thus, the expression for $\mathcal{H}$ in~\eqref{eq:H-quad2}  becomes
\begin{equation*}
\mathcal{H}=- \dfrac{1}{2}\sum_{i=1}^{n}\sum_{j \neq i} \ell^\#_{ij}(q_i - q_j)^2
\theta_j,
 \end{equation*}
establishing the result.
\end{proof}

Next, we show that $\mathcal{H}$ is non-negative despite the fact that $\ell_{ij}^\#$, for $i \neq j$, is not always a non-positive value. To achieve it, we utilize the following Lemma, which was inspired by the work of~\cite[Section 5]{deutsch1984derivatives}.

\begin{lemma}\label{thm:LU-semi-pos}
Let $L=(\ell_{ij})$ be an irreducible Laplacian matrix and $\pmb{\theta}=(\theta_1, \dots, \theta_n)^\top$ its normalized right null vector.  Then for every vector $\mathbf{y} =(y_1, \dots, y_n)^\top \in \mathbb{R}^n$,
\begin{equation}\label{LU-posi-semi}
\mathbf{y}^\top L \mathrm{diag}(\pmb{\theta}) \mathbf{y} \ge 0.
\end{equation}
\end{lemma}

\begin{proof}
Similar to the argument used in the proof of Theorem~\ref{new-H}, it can be verified that
\begin{equation*}
\mathbf{y}^\top L\mathrm{diag}(\pmb{\theta}) \mathbf{y} =  -\dfrac{1}{2} \sum_{i=1}^{n}\sum_{j \neq i}\ell_{ij} \theta_j(y_i-y_j)^2.
\end{equation*}
Since $\ell_{ij}\le 0$, for $i \neq j$, and $\theta_i>0$ (see the definitions of $L$ and $\theta$ in Section~\ref{sec:preliminary}), the right side of the above equation is non-negative, and thus $ \mathbf{y}^\top L \mathrm{diag}(\pmb{\theta}) \mathbf{y} \geq 0$.   
\end{proof}

\begin{theorem}\label{thm-Lsharp-semipos}
    Let the network heterogeneity index $\mathcal{H}$ be given as in~\eqref{H}. Then $\mathcal{H}\geq 0$,
with the equality holding if and only if $\mathbf{q}$ is a scalar multiple of vector $\mathbf{1}$.
\end{theorem}
\begin{proof}
Let vector $\mathbf{y} \in \mathbb{R}^n$ be defined as
\begin{equation}\label{eq:y}
  \mathbf{y}=\mathrm{diag}(\pmb{\theta})^{-1} L^\# \mathrm{diag}(\pmb{\theta}) \mathbf{q}.
\end{equation}
We claim that
        \begin{equation}\label{new-h1}
            \mathcal{H} = \mathbf{q}^\top L^\#\mathrm{diag}(\pmb{\theta}) \mathbf{q} = \mathbf{y}^\top \mathrm{diag}(\pmb{\theta}) L^\top \mathbf{y},
        \end{equation} 
where $L^\#$ is the group inverse of Laplacian matrix $L$. 
It follows from~\eqref{LL-sharp} and~\eqref{eq:y} that
\begin{align*}
\mathbf{y}^\top \mathrm{diag}(\pmb{\theta}) L^\top \mathbf{y}&=\mathbf{q}^\top\mathrm{diag}(\pmb{\theta}) (L^\#)^\top \mathrm{diag}(\pmb{\theta})^{-1}\mathrm{diag}(\pmb{\theta})L^\top \mathrm{diag}(\pmb{\theta})^{-1}L^\# \mathrm{diag}(\pmb{\theta}) \mathbf{q}\\
 &=\mathbf{q}^\top \mathrm{diag}(\pmb{\theta}) (LL^\#)^\top \mathrm{diag}(\pmb{\theta})^{-1}L^\#\mathrm{diag}(\pmb{\theta}) \mathbf{q}\\
     &   = \mathbf{q}^\top \mathrm{diag}(\pmb{\theta}) (I_n - \mathbf{1}\pmb{\theta}^\top) \mathrm{diag}(\pmb{\theta})^{-1}L^\#\mathrm{diag}(\pmb{\theta}) \mathbf{q}\\
        & = \mathbf{q}^\top L^\#\mathrm{diag}(\pmb{\theta}) \mathbf{q} - \mathbf{q}^\top \mathrm{diag}(\pmb{\theta}) \mathbf{1}\pmb{\theta}^\top \mathrm{diag}(\pmb{\theta})^{-1}L^\#\mathrm{diag}(\pmb{\theta}) \mathbf{q}\\
        &= \mathbf{q}^\top L^\# \mathrm{diag}(\pmb{\theta}) \mathbf{q} - \mathbf{q}^\top \mathrm{diag}(\pmb{\theta}) \mathbf{1}\mathbf{1}^\top L^\#\mathrm{diag}(\pmb{\theta})\mathbf{q}
\end{align*}
Since $\mathbf{1}^\top L^\# =0$, by Lemma~\ref{lem1}(ii), the second term on the right side vanishes, and thus~\eqref{new-h1} holds. Now, it follows  Lemma~\ref{thm:LU-semi-pos} that $\mathcal{H}\geq 0$.

It is evident from Theorem~\ref{new-H} that if $\mathbf{q}$ is a scalar multiple of vector $\mathbf{1}$, i.e., $q_i=q_j$ for $1 \leq i, j \leq n$, then $\mathcal{H}= 0$. Thus, to complete the proof, we need to show  $\mathcal{H}=0$ implies $\mathbf{q}$ is a scalar multiple of $\mathbf{1}$.
Since $\mathcal{H}=\mathbf{y}^\top \mathrm{diag}(\pmb{\theta}) L^\top \mathbf{y}=0$, by~\eqref{new-h1}, thus
$$0= \mathbf{y}^\top \mathrm{diag}(\pmb{\theta}) L^\top \mathbf{y}+ \mathbf{y}^\top L\mathrm{diag}(\pmb{\theta}) \mathbf{y}=\mathbf{y}^\top[\mathrm{diag}(\pmb{\theta}) L^\top+L\mathrm{diag}(\pmb{\theta})]\mathbf{y}.$$
    
It can be verified that $L\mathrm{diag}(\pmb{\theta})$ is an irreducible Laplacian matrix, thus $\mathrm{diag}(\pmb{\theta}) L^\top+L \mathrm{diag}(\pmb{\theta})$ becomes a symmetric irreducible  Laplacian matrix. Hence, $\mathbf{y} = \alpha \mathbf{1} \in \mathrm{ker}(\mathrm{diag}(\pmb{\theta}) L^\top+L\mathrm{diag}(\pmb{\theta}))$ for some $\alpha \in \mathbbm{R}$. It follows from~\eqref{eq:y} that
\begin{equation*}
\mathrm{diag}(\pmb{\theta})^{-1} L^\# \mathrm{diag}(\pmb{\theta}) \mathbf{q}= \alpha \mathbf{1}.
\end{equation*}
Multiplying by $L^2 \mathrm{diag}(\pmb{\theta})$ from the left side gives
$$L^2 L^\# \mathrm{diag}(\pmb{\theta})\mathbf{q}= a L^2 \mathrm{diag}(\pmb{\theta}) \mathbf{1}.$$
Applying~\eqref{groupinverse-prop} to the above equation and following $\mathrm{diag}(\pmb{\theta}) \mathbf{1}=\pmb{\theta}$ and $L \pmb{\theta}=\mathbf{0}$ result in
$$L \mathrm{diag}(\pmb{\theta}) \mathbf{q}= \alpha L^2 \pmb{\theta}= \alpha L(L \pmb{\theta})=\mathbf{0}.$$
This implies that $ \mathrm{diag}(\pmb{\theta}) \mathbf{q} \in \mathrm{ker}(L)$ is a right null vector of $L$, i.e., $\mathrm{diag}(\pmb{\theta}) \mathbf{q} = \beta \pmb{\theta}$ for some $ \beta \in \mathbbm{R}$. Hence, 
$\mathbf{q}=\beta \mathrm{diag}(\pmb{\theta})^{-1} \pmb{\theta}=\beta \mathbf{1}$ is a scalar multiple of vector $\mathbf{1}$. 
\end{proof}

The following result can be used to compare different Perron root values $r$ of different network structures and resource locations, in which they have the same network average value. Specifically, $r$ in~\eqref{spect-expand} can be rewritten in terms of $\mathcal{A}$ and $\mathcal{H}$ defined in~\eqref{A} and~\eqref{H} as 
\begin{equation}\label{Perron-A-H}
r = s(Q - \mu L) = \mathcal{A}+\frac{1}{\mu}\mathcal{H}+o\left(\frac{1}{\mu}\right).
\end{equation}
\begin{theorem}\label{thm-order-det}
Let $r^{(1)} = \mathcal{A}^{(1)} + \dfrac{1}{\mu} \mathcal{H}^{(1)}+ o\left( \dfrac{1}{\mu}\right)$ and $r^{(2)} = \mathcal{A}^{(2)} + \dfrac{1}{\mu} \mathcal{H}^{(2)}+ o\left( \dfrac{1}{\mu}\right)$. Suppose $\mathcal{A}^{(1)}=\mathcal{A}^{(2)}$
 and $\mathcal{H}^{(1)}>\mathcal{H}^{(2)}$. Then for sufficiently large value of $\mu>0$, $r^{(1)}>r^{(2)}$.
\end{theorem}
\begin{proof} 
Let $h = \mathcal{H}^{(1)} - \mathcal{H}^{(2)}>0$. Then,
$ r^{(1)} - r^{(2)} = \dfrac{1}{\mu} (\mathcal{H}^{(1)}-\mathcal{H}^{(2)})+ o\left(\dfrac{1}{\mu} \right)=\dfrac{1}{\mu}(h+o(1))>0,
$
for sufficiently large value of $\mu>0$.
\end{proof}
In the next two sections, we highlight the importance of the network heterogeneity index $\mathcal{H}$ in determining the spread of an infectious disease as well as the evolution of population of a species in heterogeneous environments of connected regions by presenting applications from epidemiology and ecology.

Throughout these examples, we assume symmetric movements between the patches. However, it is worth noticing that the network heterogeneity index $\mathcal{H}$ plays an equally important role even in networks with asymmetric movements.

\section{Application to a Multi-Patch SIS Model}\label{sec:application} 



In this section, we examine the impact of human transportation on disease transmission within a multi-patch SIS (susceptible-infected-susceptible) model. We consider two network structures, i.e., the star and path networks, and evaluate $\mathcal{H}$ under various hot spot scenarios to illustrate how different network structures and hot spot placements can affect disease outcomes.

Consider the $n$-patch SIS epidemic model ($n \ge 3$) given in~\cite{allen2007asymptotic}, in which both susceptible and infected individuals can move between the patches, as follow
 \begin{equation}\label{eq-SIS}
\left\{\begin{array}{rcl}
S_i' & = &\displaystyle - \frac{\beta_i S_i I_i}{S_i+I_i} +\gamma_i I_i + \mu_S \sum_{j=1}^n (m_{ij} S_j - m_{ji} S_i), \\
I_i' & = &\displaystyle \frac{\beta_i S_i I_i}{S_i+I_i} -\gamma_i I_i +\mu_I \sum_{j=1}^n (m_{ij} I_j - m_{ji} I_i). 
\end{array}
\right.
\quad i=1, 2, \ldots, n,
\end{equation}
Here, $S_i$ and $I_i$ denote the number of susceptible and infected individuals in patch $i$; $\beta_i\ge 0$ and $\gamma_i \ge 0$ are the disease transmission rate and recovery rate in patch $i$; $\mu_S\ge 0$ and $\mu_I\ge 0
$ denote the dispersal coefficients of the susceptible and infected individuals.~Finally, $m_{ij}\ge 0$ represents the movement weight from patch $j$ to patch $i$. We assume that the movement between the patches is symmetric. Thus, the Laplacian matrix $L$ associated to the movement, as defined in~\eqref{eq:Lap-mat}, is irreducible and symmetric with $\pmb{\theta}=(\frac{1}{n},\dots, \frac{1}{n})^\top$ being 
 its normalized right null vector. The Jacobian matrix $J$ of system~\eqref{eq-SIS} at the disease-free equilibrium is given by
\begin{equation}\label{eq:J-SIS}
J = \mathrm{diag}\{\beta_i - \gamma_i\}- \mu_I L,
\end{equation}
where $\beta_i-\gamma_i$ denotes the disease growth rate of patch $i$ when isolated from other patches, that is when $\mu_I =0$.  Patch $i$ is called a hot spot if $\beta_i - \gamma_i>0$, and a non-hotspot if $\beta_i - \gamma_i \leq 0$.
As stated in~\cite{diekmann2010construction,van2002reproduction}, the disease growth rate is determined by the sign of $s(J)$, namely, the disease dies out if $s(J)<0$, and persists if $s(J)>0$.
Since $J$ is irreducible with its off-diagonal entries being non-negative, by the Perron Frobenius theory the network disease growth rate of~\eqref{eq-SIS}, denoted as $r=s(J)$, is the Perron root of $J$. It now follows from~\eqref{spect-expand} and~\eqref{Perron-A-H} that
\begin{equation}\label{SIS-diseasegrowth}
r = \mathcal{A}+\frac{1}{\mu_I}\mathcal{H}+ o\left( \frac{1}{\mu_I}\right)=\sum_{i=1}^n \frac{(\beta_i - \gamma_i) }{n} + \frac{1}{ \mu_I }\sum_{i=1}^n \sum_{j=1}^n\frac{ (\beta_i - \gamma_i)\ell_{ij}^\# (\beta_j - \gamma_j)}{n}+ o\left( \frac{1}{\mu_I}\right),
\end{equation}
where $L^\#=(\ell^\#_{ij})$ denotes the group inverse matrix of $L$.

In what follows, we consider two distinct network structures, namely an n-patch star network and a 4-patch path network, in the SIS model in~\eqref{eq-SIS}.  To demonstrate the important role played by the network heterogeneity index $\mathcal{H}$ in determining disease outbreaks, we consider various scenarios of hot spot locations in each network structure and compare the corresponding $\mathcal{H}$ values.

\medskip
\begin{remark}\label{nonhot-hot} 
 In the following network structures, we denote the disease growth rate of hot spot patches as $q^{(h)}$ and non-hot spot ones as $q^{(\ell)}$.
\end{remark}
\medskip

The following result, which will also be used in Section~\ref{sec:application2}, formulates the network heterogeneity index $\mathcal{H}$ for networks with symmetric movements and one hot spot patch.

\begin{theorem}\label{thm-1different}
Let $L$ be an $n \times n$ irreducible and symmetric Laplacian matrix, and $Q=\mathrm{diag}({\mathbf q})=\mathrm{diag}\{q_i\}$ such that for $1\leq i \leq n$,
 $q_i =q^{(\ell)} $ if $i \neq k$, and $q_i=q^{(h)}$ if $i=k$. Then, the network heterogeneity index $\mathcal{H}$ in~\eqref{new-H} becomes
\begin{equation}\label{eq:H-one}
\mathcal{H}=\dfrac{\ell^\#_{kk}(q^{(h)} -q^{(\ell)} )^2}{n},
\end{equation}
where $L^\#=(\ell^\#_{ij})$ denotes the group inverse matrix of $L$.
\end{theorem}

\begin{proof}
Since $L$ is symmetric, $\pmb{\theta}=(\frac{1}{n}, \dots,\frac{1}{n})^\top=\frac{1}{n}\mathbf{1}$ is its normalized right null vector, and $\mathcal{H}$ in~\eqref{new-H} is given by
\begin{equation}\label{eq:H1-L-sym}
\mathcal{H}=\dfrac{1}{n}\mathbf{q}^\top L^\#\mathbf{q}. 
\end{equation}
Additionally, by the assumption, $\mathbf{q}$ can be rewritten as\begin{equation}\label{eq:q1-L-sym}\mathbf{q}=(q^{(\ell)},\dots,q^{(\ell)},\underbrace{q^{(h)}}_{\text{$k^{\mathrm{th}}$}},q^{(\ell)} \dots, q^{(\ell)})^\top
=q^{(\ell)}\mathbf{1} +(q^{(h)}-q^{(\ell)})\mathbf{e}_k,
\end{equation}
where $\mathbf{e}_k$ denotes the standard $ n \times 1$ unit vector with one in the $k^{\mathrm{th}}$ entry and zeros elsewhere, given as
\begin{equation}\label{eq:unit-vector}
\mathbf{e}_k=(0,\dots,0,\underbrace{1 }_{\text{$k^{\mathrm{th}}$}},0 \dots, 0)^\top.
\end{equation}
Substituting~\eqref{eq:q1-L-sym} into~\eqref{eq:H1-L-sym} yields
\begin{align*}
\mathcal{H}&=\dfrac{1}{n}(q^{(\ell)} \mathbf{1}^\top+(q^{(h)}-q^{(\ell)})\mathbf{e}_k^\top) L^\# ( \mathbf{1}q^{(\ell)} +\mathbf{e}_k(q^{(h)}-q^{(\ell)}))\\
&=\dfrac{1}{n}q^{(\ell)}  \mathbf{1}^\top L^\# (\mathbf{1}q^{(\ell)}  +(q^{(h)}-q^{(\ell)})\mathbf{e}_k)+\dfrac{1}{n}(q^{(h)}-q^{(\ell)})\mathbf{e}_k^\top L^\#   \mathbf{1}q^{(\ell)}\mathbf{e}_k\\
&+\dfrac{1}{n}(q^{(h)}-q^{(\ell)})\mathbf{e}_k^\top L^\# \mathbf{e}_k(q^{(h)}-q^{(\ell)}).
       \end{align*}      
\noindent Lemma~\ref{lem1}(i)-(ii) along with $\pmb{\theta}=\frac{1}{n}\mathbf{1}$, give $\mathbf{1}^\top L^\#=0$ and  $L^\#\mathbf{1}=0$. Thus, the first two terms in the above equation vanish, resulting in
$\mathcal{H}=\dfrac{1}{n}(q^{(h)}-q^{(\ell)})\mathbf{e}_k^\top L^\# \mathbf{e}_k(q^{(h)}-q^{(\ell)})=\dfrac{\ell^\#_{kk}(q^{(h)} -q^{(\ell)})^2}{n}.$
\end{proof}

\subsection{Star network with one hot spot} 
Consider the SIS model given in~\eqref{eq-SIS} over an $n$-patch star network, for $n\geq 3$, with patch $1$ at the hub and patches $2,3, \dots, n$ on the leaves as shown in Figure~\ref{fig:star1}(\subref{star}). Assume the weight between the hub and each leaf is $1$. The corresponding Laplacian matrix $L$ is shown in Figure~\ref{fig:star1}(\subref{lap-star}).
\begin{figure}[H]
 \centering
  \begin{subfigure}[b]{0.4\textwidth}
  \centering
  \begin{tikzpicture}
\begin{scope}[every node/.style={circle,  thick,draw}]
    \node [circle, fill= white](A) at (0,0) {$1$};
    \node [circle, fill=white](B) at (1.2,0) {$2$};
    \node [circle, fill=white](C) at (0,-1.2) {$4$};
    \node [circle, fill=white](D) at (-1.2,0) {$6$};
    \node [draw=white](E) at (0,1.2) {$\cdots$};
     \node [circle, fill=white](F) at (0.9,0.9) {$n$};
      \node [circle, fill=white](G) at (0.9,-0.9) {$3$};
       \node [circle, fill=white](H) at (-0.9,0.9) {$7$};
        \node [circle, fill=white](I) at (-0.9,-0.9) {$5$};
    \draw [<->, thick, draw= blue](A) -- (B);
    \draw [<->, thick, draw= blue](A) -- (C);
    \draw [<->, thick, draw= blue](A) -- (D);
    \draw [<->, thick, draw= blue](A) -- (E);
    \draw [<->, thick, draw= blue](A) -- (F);
    \draw [<->, thick, draw= blue](A) -- (G);
    \draw [<->, thick, draw= blue](A) -- (H);
    \draw [<->, thick, draw= blue](A) -- (I);
\end{scope}
\end{tikzpicture}
  \caption{}
  \label{star}
  \end{subfigure}
  \centering
   \begin{subfigure}[b]{0.4\textwidth}
  \centering
$L = \begin{pmatrix} n-1&-1&-1&\cdots&-1\\[1pt]
-1&1&0&\cdots&0\\[1pt]
-1&0&1&\cdots&0\\[1pt]
\vdots&\vdots&\vdots&\ddots&\vdots\\[1pt]
-1&0&0&\cdots&1
\end{pmatrix}$
  \caption{}
  \label{lap-star}
  \end{subfigure}
\caption{\textbf({\subref{star})} Flow diagram of an n-patch star network. {\textbf({\subref{lap-star})}} An associated Laplacian matrix $1$.}
\label{fig:star1}
\end{figure}
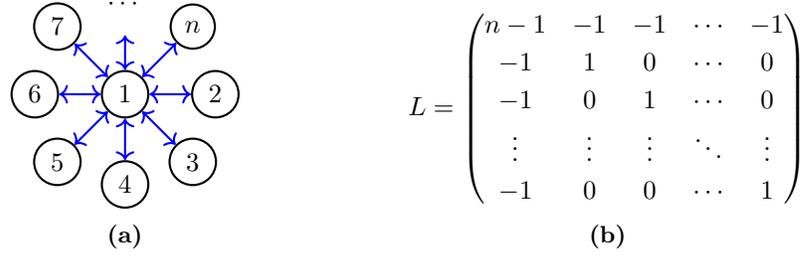
\noindent Following~\eqref{eq:Lsharp-partition}, the group inverse matrix $L^\#$ associated to $L$ is given by
\begin{equation}\label{Lsharp-star}
L^\#= \dfrac{n-1}{n^2} \mathcal{J} + \frac{1}{n}\begin{pmatrix}
0&-1&-1&\cdots&\cdots&\cdots&-1\\
-1&n-2&-2&\cdots&\cdots&\cdots&-2\\
-1&-2&n-2&-2&\cdots&\cdots&-2\\
-1&-2&-2&n-2&-2&\cdots&-2\\
\vdots&\vdots&\vdots&\vdots&\vdots&\vdots&\vdots\\
-1&-2&-2&-2&\cdots&-2&n-2\\
\end{pmatrix},
\end{equation}
where $\mathcal{J}=(1)_{n \times n}$
denotes the all-one matrix.

Next, we investigate the invasion of an infectious disease within the star network in Figure~\ref{fig:star1}(\subref{star}) by exploring two scenarios of one hot spot patch in the network: at the hub, Figure~\ref{fig:star}(\subref{star1}), and on a leaf, Figure~\ref{fig:star}(\subref{star2}).
 
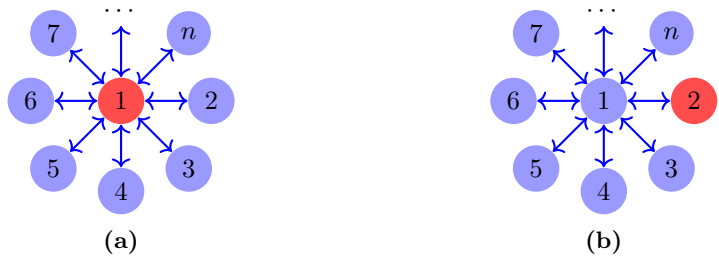
\begin{figure}[H]
 \centering
   \begin{subfigure}[t]{0.4\textwidth}
  \centering
  \begin{tikzpicture}
\begin{scope}
    \node [circle, fill= red!70](A) at (0,0) {$1$};
    \node [circle, fill=blue!40](B) at (1.2,0) {$2$};
    \node [circle, fill=blue!40](C) at (0,-1.2) {$4$};
    \node [circle, fill=blue!40](D) at (-1.2,0) {$6$};
    \node [draw=white](E) at (0,1.2) {$\cdots$};
     \node [circle, fill=blue!40](F) at (0.9,0.9) {$n$};
      \node [circle, fill=blue!40](G) at (0.9,-0.9) {$3$};
       \node [circle, fill=blue!40](H) at (-0.9,0.9) {$7$};
        \node [circle, fill=blue!40](I) at (-0.9,-0.9) {$5$};
    \draw [<->, thick, draw= blue](A) -- (B);
    \draw [<->, thick, draw= blue](A) -- (C);
    \draw [<->, thick, draw= blue](A) -- (D);
    \draw [<->, thick, draw= blue](A) -- (E);
    \draw [<->, thick, draw= blue](A) -- (F);
    \draw [<->, thick, draw= blue](A) -- (G);
    \draw [<->, thick, draw= blue](A) -- (H);
    \draw [<->, thick, draw= blue](A) -- (I);
\end{scope}
\end{tikzpicture}
  \caption{}
  \label{star1}
  \end{subfigure}
    \begin{subfigure}[t]{0.4\textwidth}
  \centering
  \begin{tikzpicture}
\begin{scope}
    \node [circle, fill= blue!40](A) at (0,0) {$1$};
    \node [circle, fill=red!70](B) at (1.2,0) {$2$};
    \node [circle, fill=blue!40](C) at (0,-1.2) {$4$};
    \node [circle, fill=blue!40](D) at (-1.2,0) {$6$};
    \node [draw=white](E) at (0,1.2) {$\cdots$};
     \node [circle, fill=blue!40](F) at (0.9,0.9) {$n$};
      \node [circle, fill=blue!40](G) at (0.9,-0.9) {$3$};
       \node [circle, fill=blue!40](H) at (-0.9,0.9) {$7$};
        \node [circle, fill=blue!40](I) at (-0.9,-0.9) {$5$};
    \draw [<->, thick, draw= blue](A) -- (B);
    \draw [<->, thick, draw= blue](A) -- (C);
    \draw [<->, thick, draw= blue](A) -- (D);
    \draw [<->, thick, draw= blue](A) -- (E);
    \draw [<->, thick, draw= blue](A) -- (F);
    \draw [<->, thick, draw= blue](A) -- (G);
    \draw [<->, thick, draw= blue](A) -- (H);
    \draw [<->, thick, draw= blue](A) -- (I);
\end{scope}
\end{tikzpicture}
  \caption{}
  \label{star2}
  \end{subfigure}
\caption{Flow diagram of the star network with the hotspot patch located~\textbf{(\subref{star1})} at the hub, and~\textbf{(\subref{star2})} on leaf $2$.}
\label{fig:star}
\end{figure}
\underline{\it Hot spot at the hub:} 
By Remark~\ref{nonhot-hot}, the disease growth rate at the hub is $\beta_1 - \gamma_1 = q^{({h})}$ and on the leaves is $\beta_2 - \gamma_2 = \dots = \beta_n - \gamma_n = q^{({\ell})}$. Applying~\eqref{eq:H-one} and~\eqref{Lsharp-star} to~\eqref{SIS-diseasegrowth} give the network disease growth rate, denoted by $r^{(1)}$, as
\begin{equation}\label{star-dis-a}
r^{(1)}=\mathcal{A}^{(1)}+\dfrac{1}{\mu_I} \mathcal{H}^{(1)} + o\left( \dfrac{1}{\mu_I} \right)=\dfrac{q^{(h)}+(n-1)q^{(\ell)}}{n}+ \dfrac{1}{\mu_I}\dfrac{n-1}{n^3}(q^{(h)}-q^{(\ell)})^2+ o\left( \dfrac{1}{\mu_I} \right).
\end{equation}

\underline{\it Hot spot on a leaf:} Without loss of generality, we assume that the hot spot is located on leaf $2$, as shown in Figure~\ref{fig:star}(\subref{star2}). Thus, $\beta_2 - \gamma_2 = q^{(h)}$ and $\beta_i - \gamma_i = q^{(\ell)}$ for $1 \leq i \neq 2  \leq n$. Similar to the previous case, the network disease growth rate, denoted by $r^{(2)}$, is given by
\begin{equation}\label{star-dis-b}
r^{(2)}=\mathcal{A}^{(2)}+\frac{1}{\mu_I} \mathcal{H}^{(2)} + o\left( \frac{1}{\mu_I} \right)=\frac{q^{(h)}+(n-1)q^{(\ell)}}{n}+ \frac{1}{\mu_I}\frac{n^2 - n-1}{n^3}(q^{(h)}-q^{(\ell)})^2+ o\left( \frac{1}{\mu_I} \right).
\end{equation}
By comparing~\eqref{star-dis-a} and~\eqref{star-dis-b}, it is clear that
$$\mathcal{A}^{(2)}=\mathcal{A}^{(1)} \quad \mathrm{and}\quad\mathcal{H}^{(2)}>\mathcal{H}^{(1)} \quad \mathrm{for} \quad n \geq 3,$$
and thus, by Theorem~\ref{thm-order-det}, $r^{(2)}>r^{(1)}$. This implies that  given a sufficiently rapid dispersion rate of infected individuals $\mu_I$ in network~\ref{fig:star1}(\subref{star}) of at least three patches, a hot spot located on a leaf patch has a greater impact on the overall disease growth rate of the network compared to a hot spot at the hub.


\subsection{Path network with one hot spot}\label{sec:4-patch} In this section, we apply the SIS model in~\eqref{eq-SIS} to a path network of four patches, as depicted in Figure~\ref{path}. For simplicity purposes, the movement weight among the patches being is assumed  to be $1$. However, it is important to note that a similar approach can be extended to any path network of n patches, for $n \ge 4$, and with asymmetric movements.
\begin{figure}[H]
 \centering
 \begin{tikzpicture}
\begin{scope}[every node/.style={thick,draw}]
    \node (A)[circle, fill= white] at (1,0) {$1$};
    \node (B)[circle, fill= white] at (2.4,0) {$2$};
    \node (C)[circle, fill= white] at (3.8,0) {$3$};
    \node (D)[circle, fill= white] at (5.2,0) {$4$};
    \draw [<->, thick, draw= blue](A) -- (B);
    \draw [<->, thick, draw= blue](B) -- (C);
    \draw [<->, thick, draw= blue](C) -- (D);
\end{scope}
\end{tikzpicture}
\hfill
  \caption{Demonstration of of a flow diagram of a 4-patch path network with symmetric movement among patches.}
  \label{path}
 \end{figure}
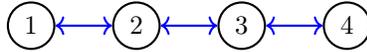
\noindent The Laplacian matrix $L$ associated to Figure~\ref{path} and its group inverse $L^\#$, by~\eqref{eq:Lsharp-partition}, are given by
\begin{equation}\label{LLsharp-path}
     L = \begin{pmatrix}
       1&-1&0&0\\
       -1&2&-1&0\\
       0&-1&2&-1\\
       0&0&-1&1
    \end{pmatrix} \quad \mathrm{and} \quad  L^\# = \frac{1}{8}\begin{pmatrix}
    7&1&-3&-5\\
    1&3&-1&-3\\
    -3&-1&3&1\\
    -5&-3&1&7
    \end{pmatrix}.
\end{equation}
\noindent Next, we explore different scenarios of one hot spot patch in network~\ref{path}. Note that the symmetric movement between the patches lead to two scenarios for the placement of one hot spot patch, as illustrated in Figure~\ref{fig:1hot-spot}(\subref{path1}) and~(\subref{path2}).
Applying~\eqref{eq:H-one} and~\eqref{LLsharp-path} to~\eqref{SIS-diseasegrowth} give the network disease growth rate of the above mentioned scenarios as
\begin{align*}
r^{(1)}&= \mathcal{A}^{(1)}+\frac{1}{\mu_I}\mathcal{H}^{(1)}+o\left(\frac{1}{\mu_I} \right)=\frac{q^{(h)}+3q^{(\ell)}}{4}+\frac{7(q^{(h)}-q^{(\ell)})^2}{32\mu_I}+o\left(\frac{1}{\mu_I} \right),\\
r^{(2)}&= \mathcal{A}^{(2)}+\frac{1}{\mu_I}\mathcal{H}^{(2)}+o\left(\frac{1}{\mu_I} \right)=\frac{q^{(h)}+3q^{(\ell)}}{4}+\frac{3(q^{(h)}-q^{(\ell)})^2}{32 \mu_I}+o\left(\frac{1}{\mu_I} \right).
\end{align*}
Clearly,  $\mathcal{A}^{(2)}=\mathcal{A}^{(1)}$ and $\mathcal{H}^{(2)}<\mathcal{H}^{(1)}$, so by Theorem~\ref{thm-order-det}, 
\begin{equation}\label{eq:r-one-hot}
r^{(2)}<r^{(1)}.
\end{equation}

\noindent Biologically, this implies that for a sufficiently rapid movement rate of infected individuals $\mu_I>0$, the placement of one hot spot in patch $1$ increases the likelihood of an outbreak in network~\ref{path} more than a hot spot in patch $2$.

 Next, we illustrate our above observations using numerical simulations. We also demonstrate that the network basic reproduction number $\mathcal{R}_0$ exhibits similar behavior to the network disease growth rate $r$. In our numerical simulations, we let the recovery rate in all four patches be the same and set $\gamma_i=1$. We choose the transmission rate of hot spot (non-hot spot) patch $i$ as $\beta_i= 1.5$ ($\beta_i =0.15$). Thus, by Remark~\ref{nonhot-hot}, $q^{(h)}=0.5$ and  $q^{(\ell)}=-0.85$.
Note that the basic reproduction number of patch $i$, for $1\leq i\leq 4$, in isolation is given $\mathcal{R}_0^{(i)}=\dfrac{\beta_i}{\gamma_i}$, see~\cite{diekmann1990definition, gao2019travel,van2002reproduction}. Thus, 
 $\mathcal{R}_0^{(h)}=1.5$ and $\mathcal{R}_0^{(\ell)}=0.15$, where $\mathcal{R}_0^{(h)}$ and $\mathcal{R}_0^{(\ell)}$ denote basic reproduction numbers for hot spot and non-hot spot patches, respectively.
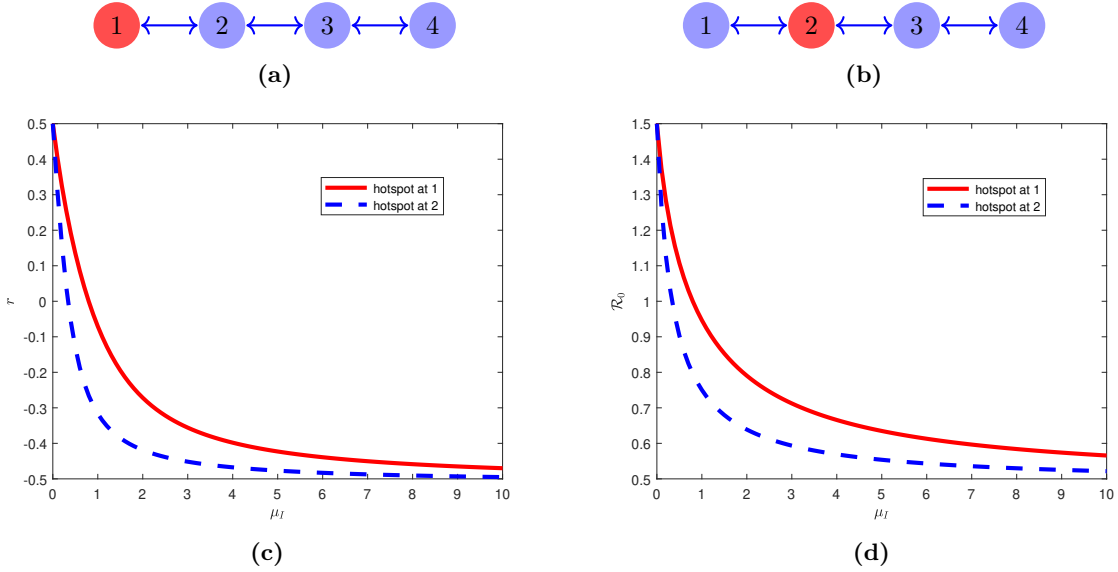
\begin{figure}[H]
 \centering
   \begin{subfigure}[t]{0.49\textwidth}
  \centering
\begin{tikzpicture}\label{one-hot-spot}
\begin{scope}
    \node (A)[circle, thick, fill= red!70] at (1,0) {$1$};
    \node (B)[circle, thick, fill= blue!40] at (2.4,0) {$2$};
    \node (C)[circle, thick, fill= blue!40] at (3.8,0) {$3$};
    \node (D)[circle, thick, fill= blue!40] at (5.2,0) {$4$};
    \draw [<->, thick, draw= blue](A) -- (B);
    \draw [<->, thick, draw= blue](B) -- (C);
    \draw [<->, thick, draw= blue](C) -- (D);
\end{scope}
\end{tikzpicture}
  \caption{}
  \label{path1}
  \end{subfigure}
    \begin{subfigure}[t]{0.49\textwidth}
  \centering
  \begin{tikzpicture}
\begin{scope}
    \node (A)[circle, thick, fill= blue!40] at (1,0) {$1$};
    \node (B)[circle, thick,  fill= red!70] at (2.4,0) {$2$};
    \node (C)[circle, thick, fill= blue!40] at (3.8,0) {$3$};
    \node (D)[circle, thick, fill= blue!40] at (5.2,0) {$4$};
    \draw [<->, thick, draw= blue](A) -- (B);
    \draw [<->, thick, draw= blue](B) -- (C);
    \draw [<->, thick, draw= blue](C) -- (D);
\end{scope}
\end{tikzpicture}
  \caption{}
  \label{path2}
  \end{subfigure}
 \begin{subfigure}[t]{0.49\textwidth}
\centering
 \includegraphics[width=\textwidth]{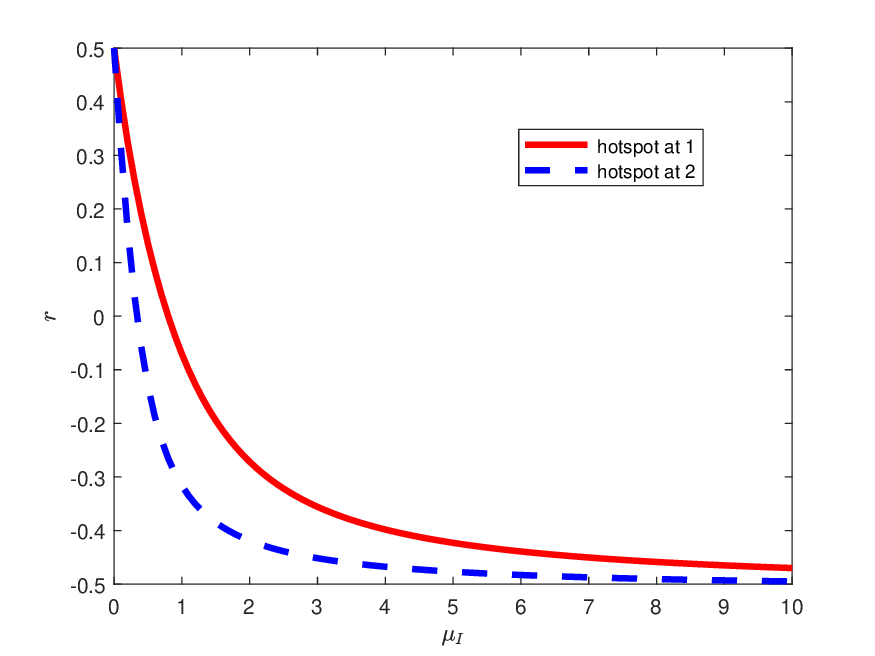} 
 \caption{}
 \label{h1r}
  \end{subfigure}%
  \hfill
  \begin{subfigure}[t]{0.49\textwidth}
  \includegraphics[width=\textwidth]{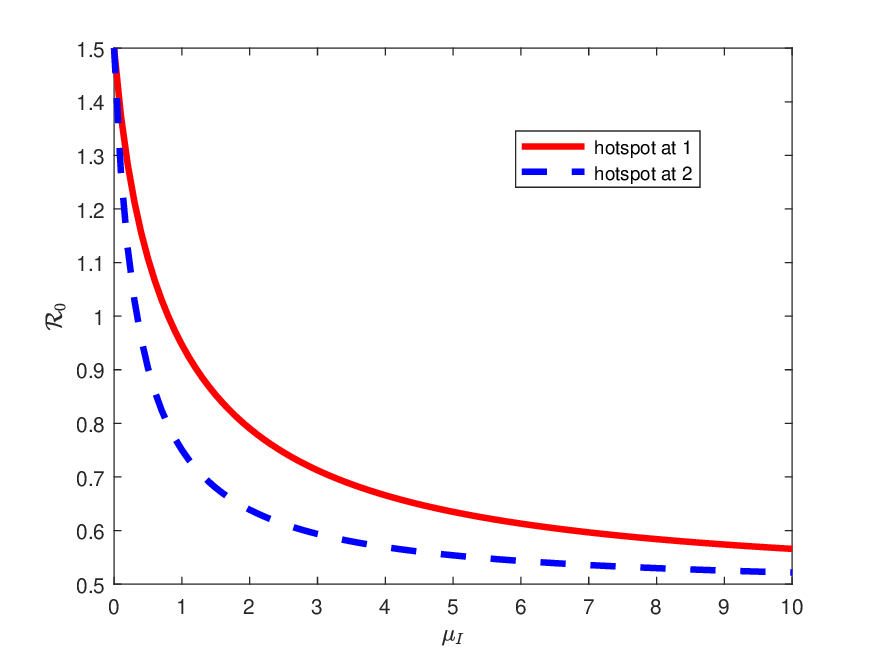}
  \caption{}
  \label{h1R0}
  \end{subfigure}%

\caption{Flow Diagrams of one hot spot in the 4-patch path network with symmetric movement on patch 1~(\subref{path1}) and path 2~(\subref{path2}). Comparison of the values of the network disease growth rates~(\subref{h1r}) and the basic reproduction numbers~(\subref{h1R0}) of both scenarios with respect to the dispersion rate $\mu_ I$.}
\label{fig:1hot-spot}
\end{figure}
As Figure~\ref{fig:1hot-spot}(\subref{h1r}) shows, the placement of one hot spot on patch $1$ leads to a higher network disease growth rate than on patch $2$, which agrees with our theoretical result in~\eqref{eq:r-one-hot}. Additionally, both $r^{(1)}$ and $r^{(2)}$ decrease and concave up as $\mu_I$ increases and $-0.5125 \leq r^{(2)} \leq r^{(1)} \leq 0.5$, aligning with the results of Lemma~\ref{mon-conv} and Theorem~\ref{them:appr-r} (by~\eqref{low-up-bound}, the exact values of lower and upper bounds are given by $\mathcal{A}=-0.5125$ and $\max_{1\leq i \leq 4}\{\beta_i -\gamma_i\}=q^{(h)}=0.5$, respectively).

Similarly, Figure~\ref{fig:1hot-spot}(\subref{h1R0}) depicts that the network basic reproduction number corresponding to the hot spot on patch $1$, $\mathcal{R}_0^{(1)}$, is higher than the one with the hot spot on patch $2$, $\mathcal{R}_0^{(2)}$, and they both decrease and concave up as $\mu_I$ increases and $0.1625\leq\mathcal{R}_0^{(2)} < \mathcal{R}_0^{(1)}\leq 1.5$. For a more comprehensive discussion on the network basic reproduction, see~\cite{allen2007asymptotic, chen2020asymptotic, gao2020fast}.

The above visualizations accurately capture the similarities in the behavior of the basic reproduction numbers and their associated disease growth rates. This can implied that, to investigate a disease outbreak in a heterogeneous network, one can study the network disease growth rate rather than the network basic reproduction number.

\subsection{Network with two hot spots} In this section, we examine the impact of different locations of the two hot spot patches in the network as depicted in Figure~\ref{path}. Specifically, we address which allocation tends to result in the largest disease growth rate. The symmetry of the network itself leads to four distinct scenarios of the arrangements of two hot spots, as shown in Figure~\ref{fig:2hot-spot2}.
\begin{figure}[H]
\centering
   \begin{subfigure}[t]{0.48\textwidth}
  \centering
\begin{tikzpicture}
\begin{scope}
    \node (A)[circle, thick, fill= red!70] at (1,0) {$1$};
    \node (B)[circle, thick, fill= red!70] at (2.4,0) {$2$};
    \node (C)[circle, thick, fill= blue!40] at (3.8,0) {$3$};
    \node (D)[circle, thick,  fill= blue!40] at (5.2,0) {$4$};
    \draw [<->, thick, draw= blue](A) -- (B);
    \draw [<->, thick, draw= blue](B) -- (C);
    \draw [<->, thick, draw= blue](C) -- (D);
\end{scope}
\end{tikzpicture}
  \caption{}
  \label{path12}
  \end{subfigure}
    \begin{subfigure}[t]{0.49\textwidth}
  \centering
\begin{tikzpicture}
\begin{scope}
    \node (A)[circle, thick, fill= blue!40] at (1,0) {$1$};
    \node (B)[circle, thick,  fill= red!70] at (2.4,0) {$2$};
    \node (C)[circle, thick, fill= red!70] at (3.8,0) {$3$};
    \node (D)[circle, thick,  fill= blue!40] at (5.2,0) {$4$};
    \draw [<->, thick, draw= blue](A) -- (B);
    \draw [<->, thick, draw= blue](B) -- (C);
    \draw [<->, thick, draw= blue](C) -- (D);
\end{scope}
\end{tikzpicture}
  \caption{}
  \label{path23}
  \end{subfigure}
 \begin{subfigure}[t]{0.49\textwidth}
\centering
\begin{tikzpicture}
\begin{scope}
    \node (A)[circle, thick,  fill= red!70] at (1,0) {$1$};
    \node (B)[circle, thick, fill= blue!40] at (2.4,0) {$2$};
    \node (C)[circle, thick, fill= red!70] at (3.8,0) {$3$};
    \node (D)[circle, thick, fill= blue!40] at (5.2,0) {$4$};
    \draw [<->, thick, draw= blue](A) -- (B);
    \draw [<->, thick, draw= blue](B) -- (C);
    \draw [<->, thick, draw= blue](C) -- (D);
\end{scope}
\end{tikzpicture}
 \caption{}
 \label{path13}
  \end{subfigure}%
  \begin{subfigure}[t]{0.49\textwidth}
  \centering
\begin{tikzpicture}
\begin{scope}
    \node (A)[circle, thick, fill= blue!40] at (1,0) {$1$};
    \node (B)[circle, thick,  fill= red!70] at (2.4,0) {$2$};
    \node (C)[circle, thick, fill= blue!40] at (3.8,0) {$3$};
    \node (D)[circle, thick, fill= red!70] at (5.2,0) {$4$};
    \draw [<->, thick, draw= blue](A) -- (B);
    \draw [<->, thick, draw= blue](B) -- (C);
    \draw [<->, thick, draw= blue](C) -- (D);
\end{scope}
\end{tikzpicture}
  \caption{}
  \label{path24}
  \end{subfigure}%
\caption{Four scenarios of two hot spot arrangements in the 4-patch path network}
\label{fig:2hot-spot2}
\end{figure}
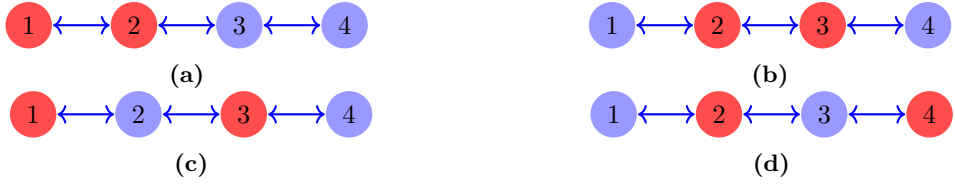
 The following result formulates the network heterogeneous network $\mathcal{H}$ of any network structure with symmetric movement and two hot spot patches. Furthermore, the result below can be extended to investigate the metapopulation growth rate of species in ecology, which is explored in Section~\ref{sec:application2}.
\begin{theorem}\label{thm-2different}
Let $L$ be an $n \times n$ symmetric and irreducible Laplacian matrix, and $Q=\mathrm{diag}(\mathbf{q})=\mathrm{diag}\{q_i\}$ such that for $1\leq i\leq n$,
$q_i =\Tilde{q}$ if $ i \neq k,\ell$, and $q_i=\check{q}$ if $ i = k,\ell$. Then,
\begin{equation}
\mathcal{H}=(\ell^\#_{kk}+\ell^\#_{k \ell }+\ell^\#_{\ell k}+\ell^\#_{\ell \ell})\dfrac{(\check{q} - \Tilde{q})^2}{n} \;,
\end{equation}
where $L^\#=(\ell^\#_{ij})$ denotes the group inverse matrix of $L$.
\end{theorem}

\begin{proof}
Following the similar arguments used in the proof of theorem~\ref{thm-1different}, we have 
$$
\mathbf{q}^\top= \Tilde{q} \mathbf{1}^\top+ (\check{q}-\Tilde{q})\mathbf{e}_k^\top +(\check{q}-\Tilde{q})\mathbf{e}_ \ell^\top,$$ where $\mathbf{e}_k^\top$ and $\mathbf{e}_\ell^\top$ are the unit vectors defined as in~\eqref{eq:unit-vector}. Substituting $\mathbf{q}^\top$ into $\mathcal{H}$ in~\eqref{eq:H1-L-sym} yields
       \begin{align*}
 \mathcal{H}&=\frac{1}{n}\mathbf{q}^\top L^\#\mathbf{q}\\
 &=\dfrac{1}{n}(\Tilde{q}\mathbf{1}^\top+(\check{q}-\Tilde{q})\mathbf{e}_k^\top+(\check{q}-\Tilde{q})\mathbf{e}_\ell^\top)L^\#(\mathbf{1}\Tilde{q}+\mathbf{e}_k(\check{q}-\Tilde{q})+\mathbf{e}_\ell(\check{q}-\Tilde{q}))\\
&=\dfrac{1}{n}\Tilde{q}\mathbf{1}^\top L^\#(\mathbf{1}\Tilde{q}+\mathbf{e}_k(\check{q}-\Tilde{q})+\mathbf{e}_\ell(\check{q}-\Tilde{q}))+\dfrac{1}{n}((\check{q}-\Tilde{q})\mathbf{e}_k^\top+(\check{q}-\Tilde{q})\mathbf{e}_\ell^\top)L^\#\mathbf{1}\Tilde{q}\\
& +\dfrac{1}{n}((\check{q}-\Tilde{q})\mathbf{e}_k^\top+(\check{q}-\Tilde{q})\mathbf{e}_\ell^\top)L^\#(\mathbf{e}_k(\check{q}-\Tilde{q})+\mathbf{e}_\ell(\check{q}-\Tilde{q})).
\end{align*}
By Lemma~\ref{lem1}(i)-(ii), $\mathbf{1}^\top L^\#=0$
 and $L^\#\mathbf{1}=0$ (since $\frac{1}{n}\mathbf{1}$ is the normalized right null vector of $L$). Thus,
$\mathcal{H}=\dfrac{1}{n}((\check{q}-\Tilde{q})\mathbf{e}_k^\top+(\check{q}-\Tilde{q})\mathbf{e}_\ell^\top)L^\#(\mathbf{e}_k(\check{q}-\Tilde{q})+\mathbf{e}_\ell(\check{q}-\Tilde{q}))=(\ell^\#_{kk}+\ell^\#_{k \ell}+\ell^\#_{\ell k}+\ell^\#_{\ell \ell})\dfrac{(\check{q} -\Tilde{q} )^2}{n}$.
\end{proof}

Following Remark~\ref{nonhot-hot}, Theorem~\ref{thm-2different}, expression~\eqref{SIS-diseasegrowth}, and the group inverse $L^\#$ in~\eqref{LLsharp-path}, the network disease growth rates of the cases in Figure~\ref{fig:2hot-spot2} are given by
\begin{align*}
    r^{(12)}&= \mathcal{A}^{(12)}+\frac{1}{\mu_I}\mathcal{H}^{(12)}+o\left(\frac{1}{\mu_I} \right)=\frac{q^{({h})}+q^{({\ell})}}{2}+\frac{3(q^{({h})}-q^{({\ell})})^2}{8\mu_I}+o\left(\frac{1}{\mu_I} \right),\\
       r^{(23)}&= \mathcal{A}^{(23)}+\frac{1}{\mu_I}\mathcal{H}^{(23)}+o\left(\frac{1}{\mu_I} \right)=\frac{q^{({h})}+q^{({\ell})}}{2}+\frac{(q^{({h})}-q^{({\ell})})^2}{8\mu_I}+o\left(\frac{1}{\mu_I} \right),\\
       r^{(13)}&= \mathcal{A}^{(13)}+\frac{1}{\mu_I}\mathcal{H}^{(13)}+o\left(\frac{1}{\mu_I} \right)=\frac{q^{({h})}+q^{({\ell})}}{2}+\frac{(q^{({h})}-q^{({\ell})})^2}{8\mu_I}+o\left(\frac{1}{\mu_I} \right),\\
       r^{(24)}&= \mathcal{A}^{(24)}+\frac{1}{\mu_I}\mathcal{H}^{(24)}+o\left(\frac{1}{\mu_I} \right)=\frac{q^{({h})}+q^{({\ell})}}{2}+\frac{(q^{({h})}-q^{({\ell})})^2}{8\mu_I}+o\left(\frac{1}{\mu_I} \right),
\end{align*}
where $\mathcal{A}^{(ij)}$ and $\mathcal{H}^{(ij)}$, for $1\leq i\neq j\leq 4$, denote the network average and network heterogeneity index when the hot spots are located on patches $i$ and $j$, respectively.
Clearly,
$\mathcal{A}^{(24)}=\mathcal{A}^{(13)}=\mathcal{A}^{(23)}=\mathcal{A}^{(12)}$ and $\mathcal{H}^{(24)}=\mathcal{H}^{(13)}= \mathcal{H}^{(23)}<\mathcal{H}^{(12)}$. Thus, by Theorem~\ref{thm-order-det}, for a sufficiently large $\mu_I$,
\begin{equation}\label{eq:2-hostspot}
  r^{(24)}, r^{(13)}, r^{(23)} < r^{(12)}.  
\end{equation}
From a biological perspective, this suggests that given a rapid movement of infected individuals, positioning hot spots in patches $1$ and $2$ simultaneously, increases the chance of a disease outbreak compared to the other three scenarios.

Next, we validate our theoretical results on the efficiency of $\mathcal{H}$ in determining a disease outbreak by conducting numerical simulations. Here, we employ the same parameter values as those used in the simulations shown in Figure~\ref{fig:1hot-spot}(\subref{h1r})-(\subref{h1R0}).
\begin{figure}[H]
 \centering
 \begin{subfigure}[t]{0.49\textwidth}
\centering
 \includegraphics[width=\textwidth]{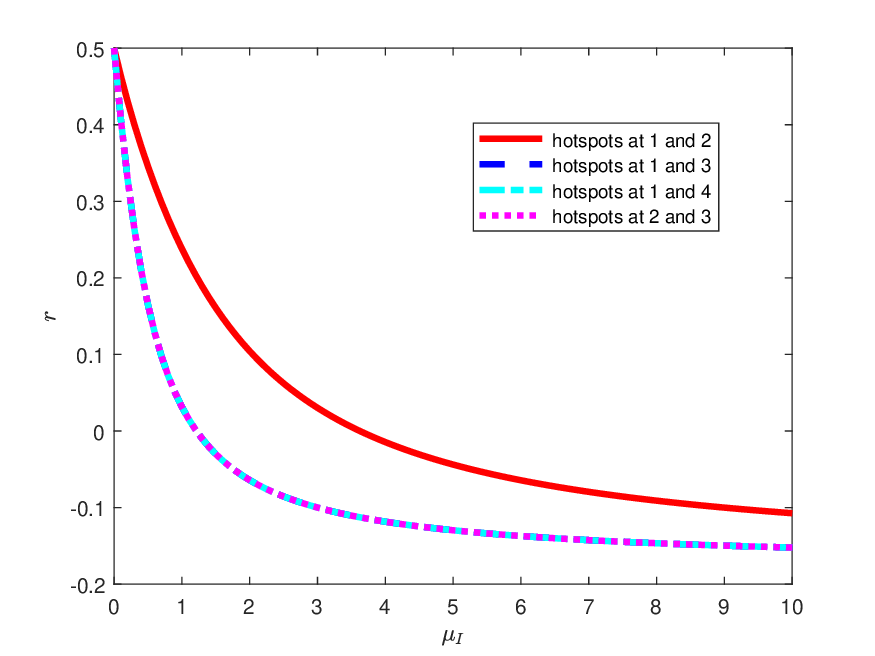} 
 \caption{}
 \label{h2r}
  \end{subfigure}%
  \hfill
  \begin{subfigure}[t]{0.49\textwidth}
  \includegraphics[width=\textwidth]{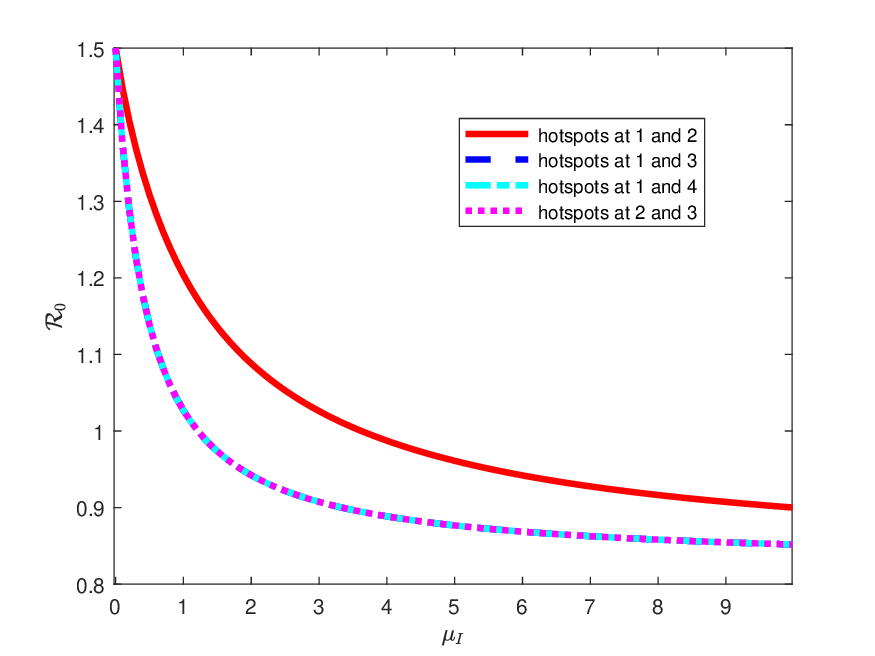}
  \caption{}
  \label{h2R0}
  \end{subfigure}%

\caption{Comparison in the behavior of~(\subref{h2r}) the network disease growth rates, and~(\subref{h2R0}) the network basic reproduction numbers   
 of four scenarios with two hot spots as $\mu_ I$ increases.}
\label{fig:2hot-spot}
\end{figure}

As Figure~\ref{fig:2hot-spot}(\subref{h2r}) illustrates, the disease growth rates of four scenarios in Figure~\ref{fig:2hot-spot2}, $r^{(12)}, r^{(23)}, r^{(13)}$ and $r^{(24)}$, decrease and concave up as $\mu_I$ increases, aligning with the result of Lemma~\ref{mon-conv}. Furthermore, $-0.175\leq r^{(24)}=r^{(13)}=r^{(23)}<r^{(12)}\leq 0.5$, leading to the following observations: First, as expected from \eqref{eq:2-hostspot}, positioning the two hot spots on patches $1$ and $2$, Figure~\ref{fig:2hot-spot2}(\subref{path12}), have the greatest impact on the network disease outbreak. Second, the other three scenarios, Figure~\ref{fig:2hot-spot2}(\subref{path23})-(\subref{path24}), result in the exactly same disease growth rate. Finally, the sharp bounds for the disease growth rates of the scenarios satisfy the sharp bounds provided in Theorem~\ref{thm:main} (by~\eqref{low-up-bound}, the lower and upper bounds are $\mathcal{A}=-0.175$ and $q^{({h})}=0.5$, respectively).
Figure~\ref{fig:2hot-spot}(\subref{h2R0}) shows that the network basic reproduction numbers of the scenarios in Figure~\ref{fig:2hot-spot2} follow a similar pattern to their corresponding disease growth rates.

In the following section, we illustrate the practicality of the network heterogeneity index $\mathcal{H}$ beyond investigating disease outbreaks in heterogeneous environments. Specifically,
 we demonstrate that $\mathcal{H}$ can also serve as a valuable tool in examining the evolution of species' population in such environments.

\section{Application to a Multi-Patch Single Species Model}\label{sec:application2}
In this section, we investigate the metapopulation growth rate of a single species model over two different network structures. Similar to Section~\ref{sec:application}, we compute the value $\mathcal{H}$  under various scenarios of patch resources to highlight the role of $\mathcal{H}$ in understanding the metapopulation growth rate.

Consider a single-species model in a 
 heterogeneous landscape of $n$ patches ($n\ge 2$), in which the species can move freely between the patches. Denote $x_i =x_i(t) \geq 0$, $1\leq i \leq n$, as the population size of the species at time $t \geq 0$ and $f_i(x_i)$ as the population growth rate on patch $i$. Let $\mu\ge 0$ be the species' diffusion coefficient and
$m_{ij}\ge 0$ denote the movement weight from patch $j$ to patch $i$. The system describing this metapopulation dynamics take the following form

\begin{equation}\label{sing-spec}
    x_i' = \dfrac{dx_i}{dt}=x_i f_i(x_i) + \mu\sum_{j=1}^n (m_{ij} x_j - m_{ji}x_i), \quad i=1,2,\ldots, n.
\end{equation}
$L$ denote the Laplacian matrix corresponding to the movement in the network. Following the assumption, $L$ is symmetric, irreducible, and singular M-matrix, with $\pmb{\theta}=(\frac{1}{n}, \dots, \frac{1}{n})^\top$ denoting its normalized right null vector.
The Jacobian matrix $J$ for the linearization~\eqref{sing-spec} at the trivial equilibrium is given by
\begin{equation}\label{J0-single-species}
   J = \mathrm{diag} \{ f_i(0)\}- \mu L.
\end{equation}
Here, patch $i$ is called a sink if $f_i(0) \leq 0$, and a source if $f_i(0)>0$.
As mentioned in Section~\ref{sec:introduction}, the metapopulation population growth rate is determined by the sign of $r=s(J)$, i.e., the population survives when $r>0$ and goes extinct when $r<0$. It follows from the assumptions on $L$ that $J$ in~\eqref{J0-single-species} is irreducible with its off-diagonal entries being non-negative. Thus, following the Perron Frobenius theory, $r$ is the Perron root of $J$. Hence, by~\eqref{spect-expand},~\eqref{eq-H2}, and~\eqref{Perron-A-H}, $r$ is given by
\begin{equation}\label{Single-popgrowth}
    r = \mathcal{A}+\frac{1}{\mu}\mathcal{H}+ o\left( \frac{1}{\mu}\right)=\frac{1}{n}\sum_{i=1}^n f_i(0)  + \frac{1}{ \mu n}\sum_{i=1}^n \sum_{j \ne i}\ell_{ij}^\# (f_i(0)-f_j(0))^2  + o\left( \frac{1}{\mu}\right),
\end{equation}
where $L^\#=(\ell^\#_{ij})$ denotes the group inverse of $L$.

In what follows, we illustrate the importance of the network heterogeneity index $\mathcal{H}$ on the metapopulation growth rate of the system~\eqref{sing-spec} by investigating two network structures. 
 \subsection{A bridge connecting two heterogeneous landscapes}\label{subsec:7path} Consider heterogeneous landscapes, compromising of $3$ and $4$ regions, located on opposite banks of a river, linked together by a bridge, as shown in Figure~(\ref{7path}).
\begin{figure}[H]
 \centering
\begin{tikzpicture}[scale=0.8,line width=1pt]
\begin{scope}[every node/.style={thick,draw}]
    \node [circle](A) at (0.1,0) {$3$};
    \node [circle](B) at (-1.2,0.9) {$1$};
    \node [circle](C) at (-1.2,-0.9) {$2$};
    \node [circle](D) at (1.9,0) {$4$};
    \node [circle](E) at (3,0.9) {$5$};
    \node [circle](F) at (3,-0.9) {$6$};
    \node [circle](G) at (4.1,0) {$7$};
   
    \draw [<->, draw= Green](A) -- (B);
    \draw [<->,  thick, draw= Green](A) -- (C);
    \draw [<->, thick, draw= Green](B) -- (C);
    \draw[<->, thick, draw= Green](A) -- (D);
    \draw [<->, thick, draw= Green](D) -- (E);
    \draw [<->, thick, draw= Green](D) -- (F);
    \draw [<->, thick, draw= Green](E) -- (G);
    \draw [<->, thick, draw= Green](F) -- (G);
\end{scope}
\end{tikzpicture}
\hfill
  \caption{Configuration of two heterogeneous environments connected via a bridge}
  \label{7path}
 \end{figure}
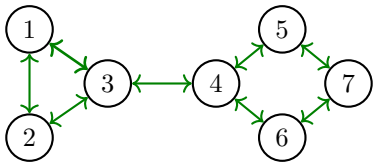
\noindent By~\eqref{eq:Lsharp-partition}, the group inverse $L^\#$ of the Laplacian matrix $L$ corresponding to the symmetric movement with weight one is given by
\begin{equation}\label{Lsharp-7patch}
L^\#=(\ell^\#_{ij})=\frac{1}{588}
\begin{pmatrix}
500 & 304 & 192 & -144 & -270 &-270 & -312\\
304 & 500 & 192 & -144 & -270 &-270 & -312\\
192 & 192 & 276 & -60 & -186 & -186 & -228\\
-144 & -144 & -60 & 192 & 66 & 66 & 24\\
-270 & -270 & -186 & 66 & 381 & 87 & 192\\
-270 & -270 & -186 & 66  & 87 & 381 & 192\\
-312 & -312 & -228 & 24 & 192 & 192 & 444
\end{pmatrix}.
\end{equation}
Note that as the bridge serves as a bottleneck that connects the two networks, $\ell^\#_{ij}$ values within the $3 \times 3$
and $4 \times 4$ blocks are notably large and positive, emphasizing the strong connection within each network. On the other hand, $\ell^\#_{ij}$ are large and negative outside of these blocks, indicating the weak connectivity between the regions of two networks. 

Next, we investigate the metapopulation growth rate of the network in configuration~\ref{7path} under three different source-sink scenarios: one source patch, two source patches, and three source patches.
\begin{remark}\label{lab:sou-sin}
    In all the scenarios below, we assume the population growth rate of the source patches as $q^{(p)}=1$ (colored in green) and sink patches as $q^{(n)}=-1$ (colored in yellow). 
\end{remark}

\underline{\it One-source scenario:} 
The symmetric movement among the patches results in five distinct scenarios for locating one source patch in the network, as patches $1$ and $2$ as well as patches $4$ and $5$ can be considered interchangeably. It follows from~\eqref{Single-popgrowth} that the network average of all the one-source patch scenarios are equal to
${\small{\mathcal{A}=\dfrac{1}{7}(q^{(p)}+6q^{(n)})=-\dfrac{5}{7}}}$.
Moreover, by Theorem~\ref{thm-1different}, the network heterogeneity index of each scenario is given by \begin{equation}\label{h:7patch1}\mathcal{H}^{(k)}={\small{\dfrac{\ell^\#_{kk}(q^{(p)}-q^{(n)})^2}{7}=\dfrac{4\ell^\#_{kk}}{7}}},
\end{equation} where $k$ denotes the location of the source patch in each scenario. Substituting $\ell_{kk}^\#$ from~\eqref{Lsharp-7patch} into~\eqref{h:7patch1} results in $\mathcal{H}^{(4)}<\mathcal{H}^{(3)}<\mathcal{H}^{(5)}<\mathcal{H}^{(7)}<\mathcal{H}^{(1)}$ with their exact values shown in Figure exact values are shown in Figure~\ref{fig:7path1}. Thus, by Theorem~\ref{thm-order-det}, $r^{(4)}<r^{(3)}<r^{(5)}<r^{(7)}<r^{(1)}$.
\begin{figure}[H]
 \centering
   \begin{subfigure}[t]{0.49\textwidth}
  \centering
\begin{tikzpicture}[scale=0.67,line width=1pt]
\begin{scope}
\node(n)[circle, fill= white,node/.style={white}] at (-3.5,0) {{\small$\mathcal{H}^{(1)}=\dfrac{500}{1029}$}};
    \node [circle,fill=orange!40!yellow](A) at (0,0) {$3$};
    \node [circle, fill=green!80](B) at (-1.3,1) {$1$};
    \node [circle, fill=orange!40!yellow](C) at (-1.3,-1) {$2$};
    \node [circle, fill=orange!40!yellow](D) at (2,0) {$4$};
    \node [circle, fill=orange!40!yellow](E) at (3.1,1) {$5$};
    \node [circle, fill=orange!40!yellow](F) at (3.1,-1) {$6$};
    \node [circle, fill=orange!40!yellow](G) at (4.2,0) {$7$};
     \node [draw=white] at (1,-2) {};
    \draw [<->, draw= Green](A) -- (B);
    \draw [<->,  draw= Green](A) -- (C);
    \draw [<->, draw= Green](B) -- (C);
    \draw[<->,  draw= Green](A) -- (D);
    \draw [<->,  draw= Green](D) -- (E);
    \draw [<->, draw= Green](D) -- (F);
    \draw [<->,  draw= Green](E) -- (G);
    \draw [<->, draw= Green](F) -- (G);
\end{scope}
\end{tikzpicture}
  \label{7path1}
  \end{subfigure}
    \begin{subfigure}[t]{0.49\textwidth}
  \centering
\begin{tikzpicture}[scale=0.67,line width=1pt]
\begin{scope}
\node(n)[circle, fill= white,node/.style={white}] at (6.5,0) {{\small$\mathcal{H}^{(7)}=\dfrac{444}{1029}$}};
    \node [circle, fill=orange!40!yellow](A) at (0,0) {$3$};
    \node [circle,fill=orange!40!yellow](B) at (-1.3,1) {$1$};
    \node [circle,fill=orange!40!yellow](C) at (-1.3,-1) {$2$};
    \node [circle,fill=orange!40!yellow](D) at (2,0) {$4$};
    \node [circle,fill=orange!40!yellow](E) at (3.1,1) {$5$};
    \node [circle,fill=orange!40!yellow](F) at (3.1,-1) {$6$};
    \node [circle,fill=green!80](G) at (4.2,0) {$7$};
    \node [draw=white] at (1,-2) {};
    \draw [<->, draw= Green](A) -- (B);
    \draw [<->,  draw= Green](A) -- (C);
    \draw [<->,  draw= Green](B) -- (C);
    \draw[<->,  draw= Green](A) -- (D);
    \draw [<->,  draw= Green](D) -- (E);
    \draw [<->,  draw= Green](D) -- (F);
    \draw [<->,  draw= Green](E) -- (G);
    \draw [<->,  draw= Green](F) -- (G);
\end{scope}
\end{tikzpicture}
  \label{7path7}
  \end{subfigure}
 \begin{subfigure}[t]{0.49\textwidth}
\centering
\begin{tikzpicture}[scale=0.67,line width=1pt]
\begin{scope}
\node(n)[circle, fill= white,node/.style={white}] at (-3.5,0) {{\small$\mathcal{H}^{(5)}=\dfrac{381}{1029}$}};
    \node [circle, fill=orange!40!yellow](A) at (0,0) {$3$};
    \node [circle,fill=orange!40!yellow](B) at (-1.3,1) {$1$};
    \node [circle,fill=orange!40!yellow](C) at (-1.3,-1) {$2$};
    \node [circle,fill=orange!40!yellow](D) at (2,0) {$4$};
    \node [circle,fill=green!80](E) at (3.1,1) {$5$};
    \node [circle,fill=orange!40!yellow](F) at (3.1,-1) {$6$};
    \node [circle,fill=orange!40!yellow](G) at (4.2,0) {$7$};
    \node [draw=white] at (1,-2) {};
    \draw [<->, draw= Green](A) -- (B);
    \draw [<->,  draw= Green](A) -- (C);
    \draw [<->, draw= Green](B) -- (C);
    \draw[<->, draw= Green](A) -- (D);
    \draw [<->,  draw= Green](D) -- (E);
    \draw [<->, draw= Green](D) -- (F);
    \draw [<->,  draw= Green](E) -- (G);
    \draw [<->, draw= Green](F) -- (G);
\end{scope}
\end{tikzpicture}
 \label{7path5}
  \end{subfigure}%
 \begin{subfigure}[t]{0.49\textwidth}
     \centering
     \begin{tikzpicture}[scale=0.67,line width=1pt]
\begin{scope}
\node(n)[circle, fill= white,node/.style={white}] at (6.5,0) {{\small$\mathcal{H}^{(3)}=\dfrac{276}{1029}$}};
    \node [circle, fill=green!80](A) at (0,0) {$3$};
    \node [circle, fill=orange!40!yellow](B) at (-1.3,1) {$1$};
    \node [circle, fill=orange!40!yellow](C) at (-1.3,-1) {$2$};
    \node [circle, fill=orange!40!yellow](D) at (2,0) {$4$};
    \node [circle, fill=orange!40!yellow](E) at (3.1,1) {$5$};
    \node [circle, fill=orange!40!yellow](F) at (3.1,-1) {$6$};
    \node [circle, fill=orange!40!yellow](G) at (4.2,0) {$7$};
    \node [draw=white] at (1,-2) {};
    \draw [<->, draw= Green](A) -- (B);
    \draw [<->,  draw= Green](A) -- (C);
    \draw [<->,  draw= Green](B) -- (C);
    \draw[<->,  draw= Green](A) -- (D);
    \draw [<->, draw= Green](D) -- (E);
    \draw [<->,  draw= Green](D) -- (F);
    \draw [<->,  draw= Green](E) -- (G);
    \draw [<->,  draw= Green](F) -- (G);
\end{scope}
\end{tikzpicture} 
\label{}
 \end{subfigure}
  \begin{subfigure}[t]{0.5\textwidth}
  \centering
\begin{tikzpicture}[scale=0.67,line width=1pt]
\begin{scope}
\node(n)[circle, fill= white,node/.style={white}] at (1,-2) {{\small $\mathcal{H}^{(4)}=\dfrac{192}{1029}$}};
    \node [circle, fill=orange!40!yellow](A) at (0,0) {$3$};
    \node [circle,fill=orange!40!yellow](B) at (-1.3,1) {$1$};
    \node [circle,fill=orange!40!yellow](C) at (-1.3,-1) {$2$};
    \node [circle,fill=green!80](D) at (2,0) {$4$};
    \node [circle,fill=orange!40!yellow](E) at (3.1,1) {$5$};
    \node [circle,fill=orange!40!yellow](F) at (3.1,-1) {$6$};
    \node [circle,fill=orange!40!yellow](G) at (4.2,0) {$7$};
    \draw [<->, draw= Green](A) -- (B);
    \draw [<->, draw= Green](A) -- (C);
    \draw [<->, draw= Green](B) -- (C);
    \draw[<->, draw= Green](A) -- (D);
    \draw [<->, draw= Green](D) -- (E);
    \draw [<->, draw= Green](D) -- (F);
    \draw [<->, draw= Green](E) -- (G);
    \draw [<->, draw= Green](F) -- (G);
\end{scope}
\end{tikzpicture}
  \label{7path4}
  \end{subfigure}%
\vskip -20pt
\caption{Illustration of $5$ scenarios with $1$ source patch (colored in green) and their corresponding network heterogeneity index values.}
\label{fig:7path1}
\end{figure}
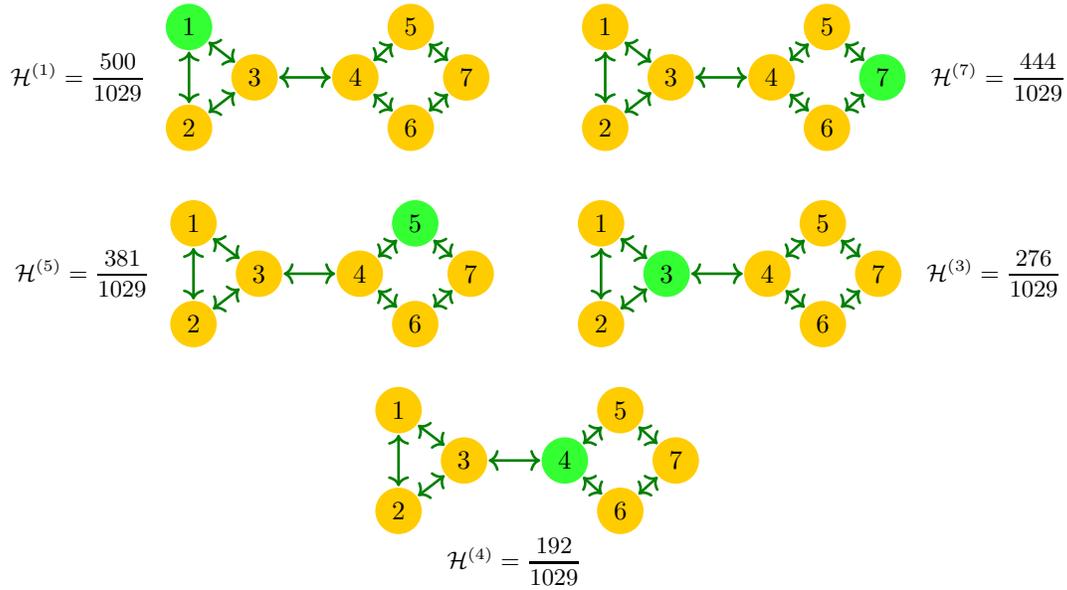
From a biological stand point, this suggest that the placement of a source patch at a corner (patches $1$, $7$ and $5$) can facilitate the metapopulation growth rate. Conversely, a source patch located closer to the bridge (patches $3$ and $4$) can inhibit the metapopulation growth rate.

\medskip
\underline{\it Two-source scenario:} The symmetric movement gives rise to twelve distinct scenarios of locating two source patches in the network, as shown in Figure~\ref{7path}. By~\eqref{Single-popgrowth}, the network average $\mathcal{A}$ in all the scenarios is the same and equal to $-\frac{3}{7}$. Following Theorem~\ref{thm-2different}, the network heterogeneity index for each scenario is given by $\mathcal{H}^{(k\ell)}=\frac{4}{7}(\ell^\#_{kk}+\ell^\#_{k \ell}+\ell^\#_{\ell k}+\ell^\#_{\ell \ell})$, where $k$ and $\ell$ denote the locations of the source patches. By~\eqref{Lsharp-7patch}, the network heterogeneity index attains its highest value when the source patches are clustered together on patches $1$ and $2$ (Figure~\ref{fig:7path23} top left), while it reaches its lowest value when the source patches are located on patches $3$ and $7$ (Figure~\ref{fig:7path23} top right). Thus, following Theorem~\eqref{thm-order-det}, the metapopulation growth rates align with their respective network heterogeneity values. That is, while the species population thrives the most when the source patches are clustered together on one side of the bridge, it is at the lowest when the source patches are located on either side of the bridge.
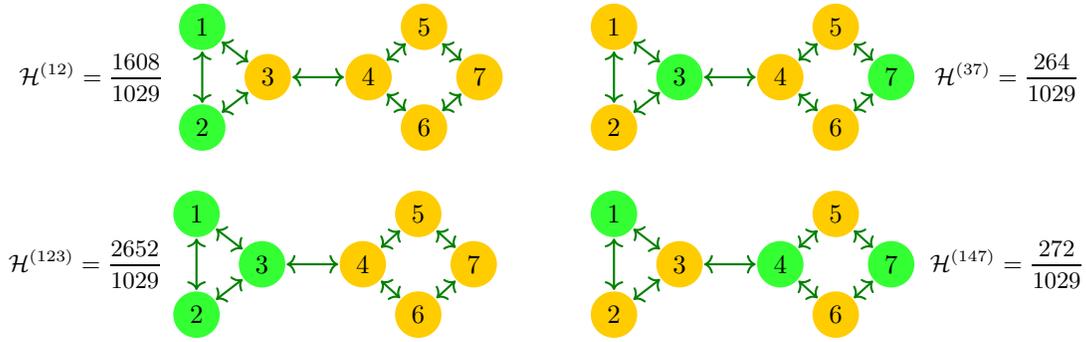
\begin{figure}[H]
 \centering
   \begin{subfigure}[t]{0.49\textwidth}
  \centering
  \begin{tikzpicture}[scale=0.67,line width=1pt]
\begin{scope}
\node(n)[circle, fill= white,node/.style={white}] at (-3.5,0) {{\small{$\mathcal{H}^{(12)}=\dfrac{1608}{1029}$}}};
    \node [circle, fill=orange!40!yellow](A) at (0,0) {$3$};
    \node [circle,fill=green!80](B) at (-1.3,1) {$1$};
    \node [circle,fill=green!80](C) at (-1.3,-1) {$2$};
    \node [circle,fill=orange!40!yellow](D) at (2,0) {$4$};
    \node [circle,fill=orange!40!yellow](E) at (3.1,1) {$5$};
    \node [circle,fill=orange!40!yellow](F) at (3.1,-1) {$6$};
    \node [circle,fill=orange!40!yellow](G) at (4.2,0) {$7$};
   \node [draw=white] at (1,-2) {};
    \draw [<->, draw= Green,thick](A) -- (B);
    \draw [<->,  thick, draw= Green](A) -- (C);
    \draw [<->, thick, draw= Green](B) -- (C);
    \draw[<->, thick, draw= Green](A) -- (D);
    \draw [<->, thick, draw= Green](D) -- (E);
    \draw [<->, thick, draw= Green](D) -- (F);
    \draw [<->, thick, draw= Green](E) -- (G);
    \draw [<->, thick, draw= Green](F) -- (G);
\end{scope}
\end{tikzpicture}
  \label{7path1}
  \end{subfigure}
    \begin{subfigure}[t]{0.49\textwidth}
  \centering
\begin{tikzpicture}[scale=0.67,line width=1pt]
\begin{scope}
\node(n)[circle, fill= white,node/.style={white}] at (6.5,0) {{\small{$\mathcal{H}^{(37)}=\dfrac{264}{1029}$}}};
    \node [circle, fill=green!80](A) at (0,0) {$3$};
    \node [circle,fill=orange!40!yellow](B) at (-1.3,1) {$1$};
    \node [circle,fill=orange!40!yellow](C) at (-1.3,-1) {$2$};
    \node [circle,fill=orange!40!yellow](D) at (2,0) {$4$};
    \node [circle,fill=orange!40!yellow](E) at (3.1,1) {$5$};
    \node [circle,fill=orange!40!yellow](F) at (3.1,-1) {$6$};
    \node [circle,fill=green!80](G) at (4.2,0) {$7$};
    \node [draw=white] at (1,-2) {};
    \draw [<->, draw= Green,thick](A) -- (B);
    \draw [<->,  thick, draw= Green](A) -- (C);
    \draw [<->, thick, draw= Green](B) -- (C);
    \draw [<->, thick, draw= Green](A) -- (D);
    \draw [<->, thick, draw= Green](D) -- (E);
    \draw [<->, thick, draw= Green](D) -- (F);
    \draw [<->, thick, draw= Green](E) -- (G);
    \draw [<->, thick, draw= Green](F) -- (G);
\end{scope}
\end{tikzpicture}
  \label{7path7}
  \end{subfigure}
 \begin{subfigure}[t]{0.49\textwidth}
\centering
  \begin{tikzpicture}[scale=0.67,line width=1pt]
\begin{scope}
    \node [circle, fill=green!80](A) at (0,0) {$3$};
    \node [circle,fill=green!80](B) at (-1.3,1) {$1$};
    \node [circle,fill=green!80](C) at (-1.3,-1) {$2$};
    \node [circle,fill=orange!40!yellow](D) at (2,0) {$4$};
    \node [circle,fill=orange!40!yellow](E) at (3.1,1) {$5$};
    \node [circle,fill=orange!40!yellow](F) at (3.1,-1) {$6$};
    \node [circle,fill=orange!40!yellow](G) at (4.2,0) {$7$};
    \node [draw=white] at (-3.5,0) {{\small{$\mathcal{H}^{(123)}=\dfrac{2652}{1029}$}}};
    \draw [<->, draw=Green,thick](A) -- (B);
    \draw [<->,  thick, draw=Green](A) -- (C);
    \draw [<->, thick, draw= Green](B) -- (C);
    \draw[<->, thick, draw= Green](A) -- (D);
    \draw [<->, thick, draw= Green](D) -- (E);
    \draw [<->, thick, draw=Green](D) -- (F);
    \draw [<->, thick, draw= Green](E) -- (G);
    \draw [<->, thick, draw= Green](F) -- (G);
\end{scope}
\end{tikzpicture}
 \label{7path5}
  \end{subfigure}%
 \begin{subfigure}[t]{0.49\textwidth}
     \centering
\label{}
 \end{subfigure}
  \begin{subfigure}[t]{0.5\textwidth}
  \centering
\begin{tikzpicture}[scale=0.67,line width=1pt]
\begin{scope}
    \node [circle, fill=orange!40!yellow](A) at (0,0) {$3$};
    \node [circle,fill=green!80](B) at (-1.3,1) {$1$};
    \node [circle,fill=orange!40!yellow](C) at (-1.3,-1) {$2$};
    \node [circle,fill=green!80](D) at (2,0) {$4$};
    \node [circle,fill=orange!40!yellow](E) at (3.1,1) {$5$};
    \node [circle,fill=orange!40!yellow](F) at (3.1,-1) {$6$};
    \node [circle,fill=green!80](G) at (4.2,0) {$7$};
    \node [draw=white] at (6.5,0) {{\small{$\mathcal{H}^{(147)}=\dfrac{272}{1029}$}}};
    \draw [<->, draw= Green,thick](A) -- (B);
    \draw [<->,  thick, draw= Green](A) -- (C);
    \draw [<->, thick, draw= Green](B) -- (C);
    \draw[<->, thick, draw= Green](A) -- (D);
    \draw [<->, thick, draw= Green](D) -- (E);
    \draw [<->, thick, draw= Green](D) -- (F);
    \draw [<->, thick, draw= Green](E) -- (G);
    \draw [<->, thick, draw= Green](F) -- (G);
\end{scope}
\end{tikzpicture}
  \label{}
  \end{subfigure}%
\caption{Demonstration of the highest and lowest values of $\mathcal{H}$ in the two-patch source scenario ({\it top plots}), and three-patch source scenario ({\it bottom plots}). }
\label{fig:7path23}
\end{figure}
\underline{\it Three-source scenario:}
The symmetric movement in the network results in the total of eighteen distinct scenarios, with the network average $\mathcal{A}$ being the same in all the scenarios, by~\eqref{Single-popgrowth}. Thus, by Theorem\ref{thm-order-det}, the highest and lowest metapopulation growth rates occur according to their network heterogeneity index values. Following~\eqref{eq-H2} and~\eqref{Lsharp-7patch}, the value of the network heterogeneity index $\mathcal{H}$ is the highest when the source patches are clustered on patches $1, 2$ and $3$, while lowest when the source patches are scattered in the network on patches $1, 4$ and $7$. As expected, clustering the source patches on one side of the bridge (Figure~\ref{fig:7path23} bottom left) yields the highest metapopulation growth rate than dispersing the sources in either side of the bridge (Figure~\ref{fig:7path23} bottom right).
\subsection{Different allocations of resources on a path network}
In this section, we investigate eight different scenarios of the met-population growth rate of the single species model in~\eqref{sing-spec} over the path network in Figure~\ref{path}. In all the scenarios, we make the following two assumptions: (i) All patches in Figure~\ref{path} only feature the source dynamic. That is, the growth rate of patch $i$, $f_i(0)>0$, for $1\leq i \leq 4$.  (ii) The network average $\mathcal{A}$ is the same and equal to $\frac{1}{4}(\sum_{i=1}^4f_i(0))=5.5$, but allocated differently among the patches. That is, the growth rate of patch $i$, $f_i(0)$, varies in quantity in different scenarios.
See Figure~\ref{fig:distA} that illustrates these phenomena. \noindent In particular, Figure~\ref{fig:distA} demonstrates that in scenarios \textit{A-G}, the population growth rate of patches range from $1$ to $10$. However, in scenario \textit{H}, the population growth rate of all patches is equal to $5.5$. 
\begin{figure}[H]
  \centering
\includegraphics[width=0.65\textwidth, height=0.45\textwidth]{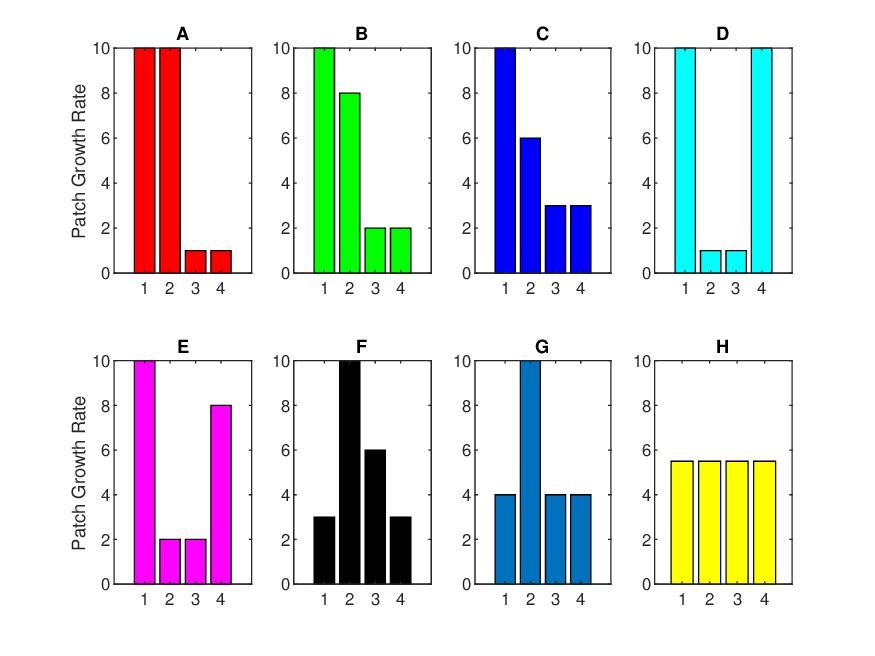}
\caption{Distribution of the network average $\mathcal{A}=5.5$ among four patches of Figure~\ref{path} in eight distinct scenarios.}
\label{fig:distA}
\end{figure}


 Figure~\ref{fig:general-case}(\subref{net-het}) depicts the order of network heterogeneity index values of the stated scenarios in Figure~\ref{fig:distA}. Additionally, it shows that while the network heterogeneity index values of scenarios $A-G$ are all positive, the value of network heterogeneity index of scenario $H$ is zero, aligning with the results of Theorem~\ref{thm-Lsharp-semipos}. The above observations are shown as 
\begin{equation}\label{eqn:orderh}
0=\mathcal{H}^{(H)}<\mathcal{H}^{(G)}<\mathcal{H}^{(F)}<\mathcal{H}^{(E)}<\mathcal{H}^{(D)}<\mathcal{H}^{(C)}<\mathcal{H}^{(B)}<\mathcal{H}^{(A)}=30.375.
\end{equation}
Figure~\ref{fig:general-case}(\subref{net-pop-grow}) illustrates the order of the metapopulation growth rates as well as their lower and upper bounds as
\begin{equation}\label{eqn:orderr}
5.5=r^{(H)} < r^{(G)} < r^{(F)} < r^{(E)} < r^{(D)} <  r^{(C)} < r^{(B)} <r^{(A)}\le 10.
\end{equation}
As expected the order of~\eqref{eqn:orderr} aligns with the order in~\eqref{eqn:orderh}, see Theorem~\ref{thm-order-det}. Furthermore, the sharp bounds in~\eqref{eqn:orderr} are in accordance with the findings of Theorem~\ref{them:appr-r}, see inequality~\eqref{low-up-bound}, where $\mathcal{A}=5.5$ and $\max_{1 \le i\le 4}\{f_i(0)\}=10$. 
\begin{figure}[H]
 \centering
  \begin{subfigure}[t]{0.49\textwidth}
  \includegraphics[width=\textwidth]{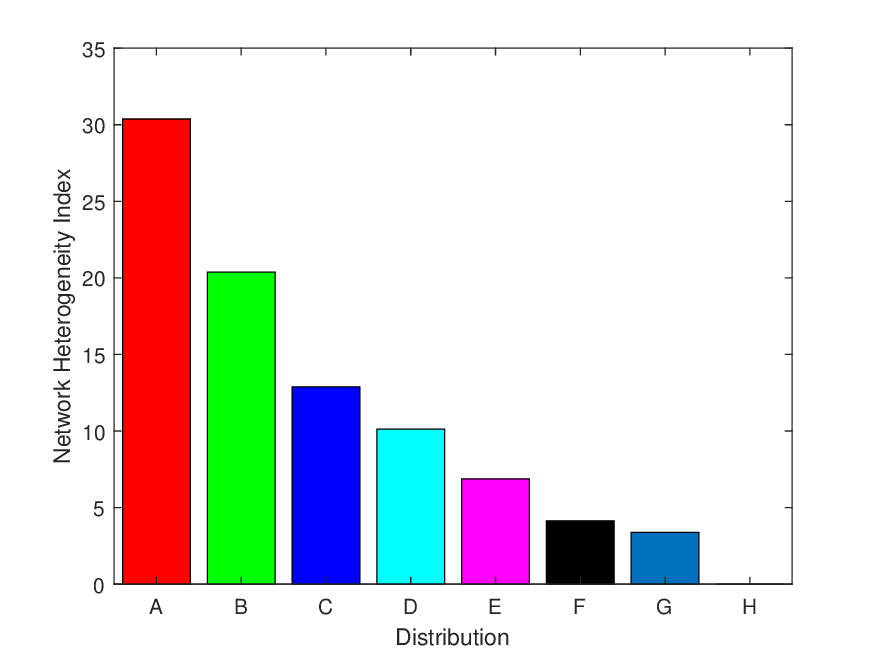}
  \caption{}
  \label{net-het}
  \end{subfigure}
  \begin{subfigure}[t]{0.49\textwidth}
  \includegraphics[width=\textwidth]{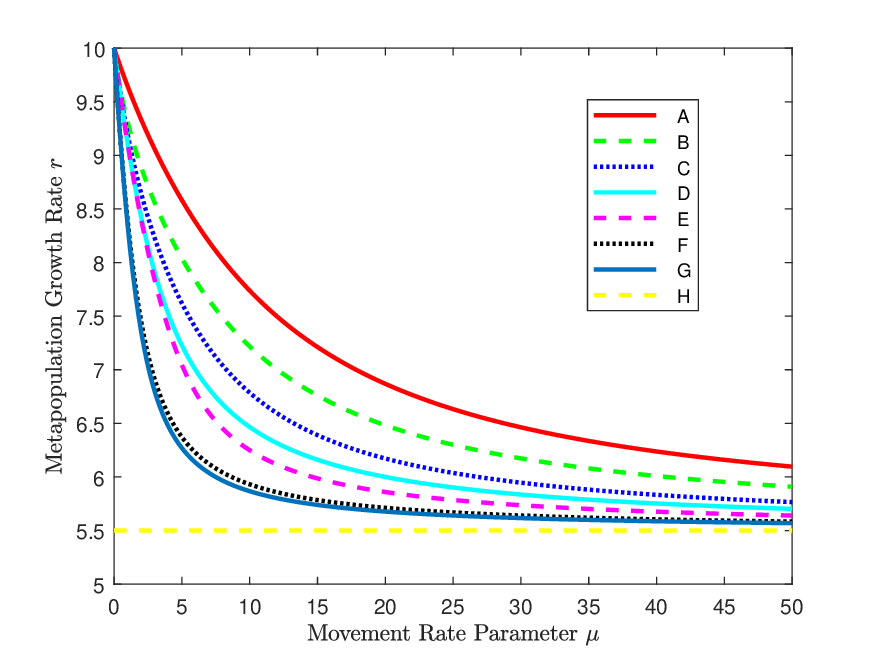}
  \caption{}
  \label{net-pop-grow}
  \end{subfigure}
\caption{\textbf{(\subref{net-het})} Comparison of the corresponding network heterogeneity indices \textbf{(\subref{net-pop-grow})} Demonstration of the behaviors of metapopulation growth rates of all scenarios}
\label{fig:general-case}
\end{figure}
Figures~\ref{fig:distA} and~\ref{fig:general-case} lead us to the following biological observations:  As the movement rate of a species over the path network (depicted in ~\ref{path}) increases, the metapopulation growth rate of all scenarios \textit{A--G} decreases, as shown in Lemma~\ref{mon-conv}, whereas for the homogeneous scenario \textit{H}, $r^{(H)}$ remains at $5.5$, consistent with Theorem~\ref{them:appr-r}(ii). Furthermore, clustering the source patches together at a corner of the network, such as the one in Figure~\ref{path}, favors the metapopulation growth rate, while distributing them inhibits it.

\section{Discussion}\label{sec:conclusion}
In recent decades, growing global connectivity and ease of travel have greatly impacted many aspects of human experience, including economies, cultures, and the exchange of ideas and information. While globalization has brought significant benefits, it has also contributed to the increased spread of infections and pathogens. For example, the swift transmission of infectious diseases like COVID-19 across borders has had a global impact on individuals and communities. The effects of globalization are not limited to human impacts. The increase in connectivity has also 
 influenced the evolutionary dynamics of other populations by altering environmental conditions and facilitating the movement of species across regions that were once isolated.

The impacts of globalization are explored using the concept of heterogeneous environments. In recent years, there have been numerous studies on the impact of the structure of heterogeneous environments, dynamics of the regions, and dispersal between them on infectious diseases, see~\cite{cantrell2020evolution, evans2013stochastic,hastings1983can,hening2018stochastic,holland2008strong,schreiber2011persistence}, as well as evolution of populations of species, see~\cite{allen2007asymptotic, arino2003multi,kirkland2021impact,gao2019travel,eisenberg2013,tien2015disease}.

Tien et al.~\cite{tien2015disease} investigated the invasiveness of a water-borne disease, like cholera, in heterogeneous environments of connected patches.  Inspired by the work of Langenhop~\cite{langenhop1971laurent}, they employed a Lauren series expansion of a perturbed irreducible Laplacian matrix to approximate the network basic reproduction number $\mathcal{R}_0$ up to the second order perturbation term. This expansion, which highlights the crucial role of the group inverse of such Laplacian matrix, serves as the motivation of this paper.

In this paper, we investigated the spread of infectious diseases in heterogeneous networks by deriving an expansion for the network disease growth rate $r$. Additionally, since $r$ can be interpreted as the metapopulation growth rate in ecological settings,  our expansion on $r$ 
 can be extended to the study of population evolution of various species.

Our theoretical analysis in Section~\ref{sec:perturbation theories} revealed that $r$ is the Perron root of the Jacobian matrix $J$ at an equilibrium point. We utilized perturbation analysis theory and the properties of group inverses of irreducible Laplacian matrices to approximate $r$ with respect to the dispersal rate $\mu$ up to the second perturbation term as
$r= \mathcal{A}+\dfrac{1}{\mu}\mathcal{H}+ o(\dfrac{1}{\mu})$,
where 
\begin{itemize}
    \item[$\mathcal{A}$] is the network average, denoting the weighted average of the dynamics of each patch in isolation where the distribution of the weights are in accordance to the network structure. 
    \item[$\mathcal{H}$] is the network heterogeneity index, capturing the network structure, the patch dynamics, and how variations in a patch reflects on the other one.
\end{itemize}

As $\mathcal{H}$ depends on the heterogeneity, network structure and the relationship between the regions, any alternation or perturbation in these parameters can change the value of $\mathcal{H}$, and thus impact $r$. In Section~\ref{sec:networkheterogenity}, we presented mathematical analysis on the network heterogeneity index $\mathcal{H}$ to enhance our understanding of the fundamental role played by $\mathcal{H}$ in comprehending the evolution of populations and spread of infectious diseases.

In Section~\ref{sec:application}, we applied the theoretical results on $\mathcal{H}$ to a multi-patch SIS (susceptible-infected-susceptible) model over two different network structures, illustrating the efficiency of $\mathcal{H}$ in understanding the spread of an infectious disease in heterogeneous environments. We considered different scenarios of the placement of disease hotspots to investigate their influence on the entire environment. Our results revealed that the clustering hotspots together at a corner of the network facilitates the disease spread, whereas scattering the hotspots throughout the network hinders its progression. Furthermore, we conducted numerical simulations to further support our theoretical analysis. In Section~\ref{sec:application2}, we extended our mathematical analysis to ecological settings by examining a multi-patch single species model over two heterogeneous network structures. We found analogous results to those observed in the study of SIS model.

Following our results, the expansion of $r$ does not impose any assumptions on heterogeneous environments beyond the connectivity of patches. Thus, it can be applied to various multi-patch epidemic models, including SIR (susceptible-infected-recovered) and SIRS (susceptible-infectious-recovered-susceptible) models, as well as SEIR (susceptible-exposed-infectious-recovered) and SIWR (susceptible-infectious-water-recovered) models in which the dispersal of infectious population is ignored. Furthermore, this expansion can be extended to multi-patch Predator-Prey models, enabling us to study the long term behavior of the predator and prey populations under specific assumptions.

\vspace{10pt}
\noindent{\bf Acknowledgments.}

This research was partially supported by the National Science Foundation (NSF) through the grant DMS-1716445. P. Yazdanbakhsh acknowledges the support of the Yvette Kanouff Industrial Mathematics Scholarship at the University of Central Florida.

\end{document}